\RequirePackage[l2tabu,orthodox]{nag}
\documentclass
[11pt,letterpaper]
{article} 

\usepackage[notes=false,later=false,camera=false]{dtrt}
\usepackage[utf8]{inputenc}
\usepackage{ stmaryrd }
\usepackage{xspace,enumerate}
\usepackage[T1]{fontenc}
\usepackage[full]{textcomp}
\usepackage[american]{babel}
\usepackage{mathtools}

\usepackage{amsthm}
\usepackage{empheq}

\hypersetup{
colorlinks=true,
urlcolor=Cerulean,
linkcolor=NavyBlue,
citecolor=PineGreen,
linktocpage=true,
}
\renewcommand*\backref[1]{\ifx#1\relax \else (pg. #1) \fi}

\usepackage[capitalise,nameinlink]{cleveref}

\crefname{lemma}{Lemma}{Lemmas}
\crefname{fact}{Fact}{Facts}
\newcommand{\colorconstraints}{\text{Color Constraints}}
\crefname{colorconstraints}{(color constraints)}{Color Constraints}
\crefformat{colorconstraints}{#2\colorconstraints#3}
\crefname{indsetconstraints}{(indset constraints)}{IndSet Constraints}
\crefformat{indsetconstraints}{#2$\mathsf{IndSet\ Axioms}$#3}
\crefname{theorem}{Theorem}{Theorems}
\crefname{mtheorem}{Theorem}{Theorems}
\crefname{itheorem}{Theorem}{Theorems}
\crefname{corollary}{Corollary}{Corollaries}
\crefname{claim}{Claim}{Claims}
\crefname{example}{Example}{Examples}
\crefname{algorithm}{Algorithm}{Algorithms}
\crefname{problem}{Problem}{Problems}
\crefname{definition}{Definition}{Definitions}
\crefname{equation}{Eq.}{Eq.}
\crefname{strategy}{Strategy}{Strategies}

\usepackage{paralist}
\usepackage{turnstile}
\usepackage[framemethod=TikZ]{mdframed}
\mdfsetup{frametitlealignment=\center}
\usepackage{tikz}
\usepackage{caption}
\DeclareCaptionType{Algorithm}
\usepackage{newfloat}

\newtheorem{theorem}{Theorem}[section]
\newtheorem{lemma}[theorem]{Lemma}
\newtheorem*{lemma*}{Lemma}

\newtheorem{proposition}[theorem]{Proposition}
\newtheorem{fact}[theorem]{Fact}

\newtheorem{question}[theorem]{Question}

\theoremstyle{definition}
\newtheorem{definition}[theorem]{Definition}
\newtheorem*{definition*}{Definition}
\newtheorem{remark}[theorem]{Remark}

\newtheorem{algorithm-thm}[theorem]{Algorithm}

\usepackage[
letterpaper,
top=1.2in,
bottom=1.2in,
left=1in,
right=1in]{geometry}
\usepackage{CormorantGaramond, inconsolata}
\usepackage[scr=rsfso]{mathalfa}
\usepackage{bm} 

\usepackage{microtype}

\usepackage{footnotebackref}

\allowdisplaybreaks
\newcommand{\FormatAuthor}[3]{
\begin{tabular}{c}
#1 \\ {\small\texttt{#2}} \\ {\small #3}
\end{tabular}
}

\newcommand{\R}{{\mathbb R}}

\newcommand{\eps}{\varepsilon}

\newcommand{\F}{{\mathbb F}}

\newcommand{\E}{{\mathbb E}}

\newcommand{\1}{\mathbf{1}}
\newcommand{\ip}[1]{\langle #1 \rangle}

\newcommand{\Z}{\mathbb Z}

\newcommand{\C}{\mathbb C}

\newcommand{\Fits}{\{\pm 1\}}

\newcommand{\cH}{\mathcal H}

\newcommand{\cA}{\mathcal A}
\newcommand{\cB}{\mathcal B}
\newcommand{\cG}{\mathcal G}

\newcommand{\poly}{\mathrm{poly}}

\newcommand{\mper}{\,.}
\newcommand{\mcom}{\,,}

\newcommand{\Id}{\operatorname{Id}}

\newcommand{\polylog}{\operatorname{polylog}}

\newcommand{\cC}{\mathcal C}
\newcommand{\cL}{\mathcal L}

\renewcommand{\emptyset}{\varnothing}
\renewcommand{\geq}{\geqslant}
\renewcommand{\leq}{\leqslant}
\renewcommand{\preceq}{\preccurlyeq}
\renewcommand{\succeq}{\succcurlyeq}
\renewcommand{\epsilon}{\varepsilon}

\newcommand{\scrR}{\mathscr{R}}
\newcommand{\tL}{\widetilde{L}}
\newcommand{\tA}{\widetilde{A}}

\newcommand{\ignore}[1]{}
\usepackage{macros}

\begin{document}
\author{
 \begin{tabular}{cc}
 \FormatAuthor{Arpon Basu\thanks{Supported by the Moskewicz Venture Forward Graduate Fellowship and NSF CAREER Award \#2047933.}}{arpon.basu@princeton.edu}{Princeton University} &
 \FormatAuthor{Pravesh K. Kothari\thanks{Supported by NSF CAREER Award \#2047933, Alfred P. Sloan Fellowship and a Google Research Scholar Award.}}{kothari@cs.princeton.edu}{Princeton University} \\
  & \\
  \FormatAuthor{Yang P. Liu\thanks{Part of this work was conducted while YL was a Postdoctoral Member at the Institute for Advanced Study. This material is based upon work supported by the National Science Foundation under Grant No. DMS-1926686.}}{yangl7@andrew.cmu.edu}{Carnegie Mellon University} &
  \FormatAuthor{Raghu Meka\thanks{Supported by NSF EnCORE: Institute for Emerging CORE Methods in Data Science Award \#2217033 and NSF AF: Small Award \#2425350}}{raghum@cs.ucla.edu}{University of California, Los Angeles} 
  \end{tabular}
  }
\title{Sparsifying Sums of Positive Semidefinite Matrices}
\maketitle

\vspace{-1em}

\begin{abstract}
In this paper, we revisit spectral sparsification for sums of arbitrary positive semidefinite (PSD) matrices. Concretely, for any collection of PSD matrices $\mathcal{A} = \{A_1, A_2, \ldots, A_r\} \subset \mathbb{R}^{n \times n}$, given any subset $T \subseteq [r]$, our goal is to find \emph{sparse} weights $\mu \in \mathbb{R}_{\geq 0}^r$ such that 
$$(1 - \epsilon) \sum_{i \in T} A_i \preceq \sum_{i \in T} \mu_i A_i \preceq (1 + \epsilon) \sum_{i \in T} A_i.$$
This generalizes spectral sparsification of graphs which corresponds to $\mathcal{A}$ being the set of Laplacians of edges. It also captures sparsifying Cayley graphs by choosing a subset of generators. The former has been extensively studied with optimal sparsifiers known. The latter has received attention recently and was solved for a few special groups (e.g., $\mathbb{F}_2^n$). 

Prior work~\cite{SilvaHS16,Cohen16b,Branden18} shows any sum of PSD matrices can be sparsified down to $O(n)$ elements. This bound however turns out to be too coarse and in particular yields no non-trivial bound for building Cayley sparsifiers for Cayley graphs.

In this work, we develop a new, \emph{instance-specific} (i.e., specific to a given collection $\mathcal{A}$) theory of PSD matrix sparsification based on a new parameter $N^*(\mathcal{A})$ which we call \emph{connectivity threshold} that generalizes the threshold of the number of edges required to make a graph connected. 

Our main result gives a sparsifier that uses at most $O(\epsilon^{-2}N^*(\mathcal{A}) (\log n)(\log r))$ matrices and is constructible in randomized polynomial time. We also show that we need $N^*(\mathcal{A})$ elements to sparsify for any $\epsilon < 0.99$.

As the main application of our framework, we prove that any Cayley graph can be sparsified to $O(\epsilon^{-2}\log^4 N)$ generators. Previously, a non-trivial bound on Cayley sparsifiers was known only in the case when the group is $\mathbb{F}_2^n$.

\end{abstract}

\thispagestyle{empty}
\setcounter{page}{0}

\clearpage
 \microtypesetup{protrusion=false}
  \tableofcontents{}
  \microtypesetup{protrusion=true}
\thispagestyle{empty}
\setcounter{page}{0}

\clearpage

\pagestyle{plain}
\setcounter{page}{1}

\section{Introduction}
In this work, we revisit the problem of sparsifying a sum of arbitrary positive semi-definite matrices. More precisely, fix an arbitrary base set of $n \times n$ PSD matrices $\mathcal{A} = \{A_1,\ldots,A_r\}$. For an arbitrary set $S \subseteq [r]$, we are interested in finding weights $\mu: S \rightarrow \R_{\geq 0}$ with sparse support $|\supp(\mu)| \ll r$ such that 
$$ \sum_{i \in S} \mu(i) A_i \approx_\eps \sum_{i \in S} A_i?$$

Here, for any two PSD matrices $A, B$, we write $A \approx_\eps B$ if $(1-\eps) B \preceq A \preceq (1+\eps) B$.

When $\cA= \{L_e : e = \{i,j\} \mid 1 \leq i < j \leq n\}$,\footnote{As a reminder, $L_e$ is the matrix with $L_e[i,i] = L_e[j,j] = 1$; $L_e[i,j] = L_e[j,i] = -1$ and all other entries are zero.} the set of all edge Laplacians on $[n]$, the problem becomes the well-studied spectral graph sparsification~\cite{SpielmanT11}. This approach generalizes the classical cut sparsifiers of Bencz\'ur and Karger~\cite{BenczurK96}.

Spielman and Srivastava~\cite{SpielmanS11} showed that sampling each edge with a probability proportional to its \emph{effective resistance}—a measure of its importance—yields a spectral sparsifier with $O(\epsilon^{-2}n\log n)$ edges. Later, the celebrated work of Batson, Spielman, and Srivastava~\cite{BatsonSS14} improved this bound, establishing the existence of spectral sparsifiers with only $O(\epsilon^{-2}n)$ edges. This bound's dependence on $n$ is optimal, as a connected graph requires $n-1$ edges and connectivity is a minimum necessity for sparsification if the original graph is connected. The dependence on $\epsilon$ is also tight up to a constant, due to the Alon-Boppana theorem~\cite{Nilli91} and its weighted generalizations~\cite{SrivastavaT18,PolyanskiiS2024}.

More generally, when $\cA$ is an arbitrary collection of rank-1 matrices, the problem appears in the study of matrix discrepancy. The breakthrough work of Marcus, Spielman, and Srivastava~\cite{MarcusSS15} provides \emph{unweighted} $O(n)$-sparsifiers for sums of arbitrary rank-1 matrices, resolving the Kadison-Singer problem.

These results can be readily generalized to obtain $O(n)$-size sparsifiers for sums of arbitrary, \emph{high rank}, $n\times n$ PSD matrices. Indeed, an $O(\epsilon^{-2}n\log n)$ bound can be obtained by an application of matrix concentration inequalities as in~\cite{SpielmanS11}. The additional $\log n$ factor can be removed by a non-trivial generalization of ideas in~\cite{BatsonSS14,Allen-ZhuLO15}. Since this is an important point of comparison for our work, we record a statement below. 

\begin{theorem}[Theorem 3 of~\cite{SilvaHS16}]
Given any $n \times n$ PSD matrices $A_1, A_2, \ldots$ and $\eps >0$, there exists a $O(\epsilon^{-2}n)$ sparse vector $\mu$ (computable in polynomial time) such that $\sum_i \mu_i A_i \approx_\epsilon \sum_i A_i$.  \label{thm:silva-result}
\end{theorem}

On the one hand, this bound matches, up to an absolute constant, the lower bound that holds even for rank $1$ PSD matrices. However, this result is can be far from optimum when the matrices $A_i$ have higher rank. For example, if the matrices $A_i$ are all of rank $k$, then it is possible that the true sparsity threshold is $O_\eps((n/k)\polylog(n))$ (as $n/k$ matrices suffice to overcome the bottleneck of getting to full-rank). We also note that extensions of the matrix discrepancy results of \cite{MarcusSS15} have been obtained for higher rank matrices in \cite{Cohen16a,Cohen16b,Branden18,BownikCMS19,XuXZ23,Bownik24}. However, all of these results lose factors that are dependent on the rank or analogous analytic measures (e.g., ratio of trace to spectral norm). 

This leaves open the tantalizing possibility of characterizing when it is possible to significantly beat the $O(n)$ bound by \emph{exploiting} that $A_i$'s are high-rank. As we next discuss, finding such a result directly relates to resolution of the problem of finding Cayley sparsifiers for Cayley graphs. 

\paragraph{Cayley sparsifiers for Cayley graphs} 
Recall that given a group $\cG$ and an arbitrary symmetric (i.e., $s \in S \iff s^{-1} \in S$) subset $S \subseteq \cG$, the Cayley graph $\Cay(\cG,S)$ on $\cG$ generated by $S$ is the graph with elements of $\cG$ as vertices with an edge between any $x,y\in \cG$ if $y = x \cdot s$. Cayley graphs are fundamental objects arising naturally in studying properties of groups and also in various graph theoretic constructions (e.g., expanders). Cayley graphs, like all graphs, admit $O(n)$ edge spectral sparsifiers. However, these sparsifiers, in general, may lose all critical algebraic structure. Thus, it is natural to ask if we can find spectral sparsifiers $\Cay(\cG,S)$ that are themselves Cayley graphs on $\cG$. That is, is there a sparse (weighted) sub-collection $T \subseteq S$ such that $\Cay(\cG,T) \approx_\epsilon \Cay(\cG,S)$?

For each $s \in \cG$, consider the Laplacian matrix $L_s$ of the graph~\footnote{In the technical sections, we use a rescaled version for convenience.} induced by all the edges corresponding to generators $s$ and $s^{-1}$. Then, the Laplacian of $\Cay(\cG,S)$ is clearly $\sum_{s \in S} (1/2) L_s$. Observe that $L_s$ is a disjoint union of cycles of size $\ell = \operatorname{order}(s)$, and thus the rank of $L_s$ (for $s\neq\id_{\cG}$) is $|\cG|(1 - 1/\ell)\geq|\cG|/2$. Observe that \cref{thm:silva-result} only yields a Cayley sparsifier with a trivial bound of $O(|\cG|)$ generators in this case! 

Previously, no better bounds were known for building sparsifiers of Cayley graphs for arbitrary groups.  A recent work~\cite{KhannaPS24} makes a beautiful new connection to \emph{code sparsification} and uses it to derive a Cayley sparsifier with $O((\log |\cG|)(\poly (\log \log |\cG|)))$ generators for Cayley graphs on $\cG = \F_2^n$. Their techniques appear to extend to $\F_q^n$ for any fixed $q$ (more generally, the error seems to degrade polynomially in $q$), but no further, to the best of our knowledge. Nevertheless, their result suggests that significantly better bounds may be possible for sparsifying sums of high-rank PSD matrices.

\paragraph{A Rank-Sensitive PSD Sparsification Result?} Given the discussion above, one may hope for sparsifiers for sums of PSD matrices with $o(n)$ elements when the constituent PSD matrices $A_i$ all have very high rank. Unfortunately, such hope turns out to be false. Indeed, we show: 

\begin{theorem}[See~\cref{schreiersparlowerbound} for further details] \label{thm:schreier-intro}
For every $n\geq 3$, there are $m = n$, $4n \times 4n$ PSD matrices $A_1, A_2, \ldots, A_m$ each of rank $2n$ such that for every $U \subsetneq [m]$, there is a vector $x$ such that $x^{\top} A_i x = 0$ for every $i \in U$ while $x^{\top} \sum_{i \in [m]} A_i x >0$. Thus, no (weighted) proper subset of $A_i$s can be a sparsifier for any $\epsilon <1$. 
\end{theorem}

Our example quashes the hope for non-trivial sparsification of PSD matrices, even when they are high-rank. The constituent $A_i$ matrices are adjacency matrices of perfect matchings on $4n$ vertices. Consequently, they also have a high analytic rank (the ratio of trace to spectral norm). These properties close the door on simple modifications.

\paragraph{An Instance-Specific Theory for PSD Matrix Sparsification} In this work, we develop a new \emph{instance-specific} theory for PSD matrix sparsification. Our goal is to identify a natural (and, for applications, understandable) parameter of the base set $\cA$ of PSD matrices that controls the size of sparsifier for sum of an arbitrary subset of $\cA$. One way to interpret our previous discussion is that the minimum rank of any matrix in $\cA$ is \emph{not} such a parameter. This is because when $\cA$ is the set of all Laplacians of edges induced by an element of $\cG = \F_2^m$, each of rank $|\cG|/2$, the work of~\cite{KhannaPS24} gives a $\tilde{O}(m) = \tilde{O}(\log |\cG|)$ size sparsifier while the collection of rank $\Omega(n)$ matrices from \cref{thm:schreier-intro} admit no sparsifier of size $o(n)$.

To motivate our instance-specific theory, let's reinterpret spectral graph sparsification. For graphs on $n$ vertices, any graph needs at least $n-1$ edges to be connected. And thus, for any connected graph $G$, any sparsifier must contain at least $n-1$ edges as otherwise, there's a cut (and thus the corresponding quadratic form of the Laplacian) that is non-zero in the original graph but $0$ in the sparsifier. The connectivity threshold is an easy lower bound on the size of a spectral sparsifier. It is a highly non-trivial result, but it also, up to a constant, turns out to be an upper bound on the size of a spectral sparsifier for any graph. 

We will now define a parameter that can be seen as generalizing redundancy in the connectivity of graphs. Our main result shows that the connectivity parameter and the related connectivity threshold for a collection of PSD matrices characterizes the size of the sparsifier for any subset of the collection up to a polylogarithmic factor in the dimension.

\begin{definition}[Connectivity Parameter]
\torestate{
\label{def:connparam}
    Let $\cA = \{A_1, \ldots, A_r\}$ be a collection of PSD matrices. For any $\alpha \geq 0$, the \emph{connectivity parameter} with strength $\alpha$, denoted $N(\alpha; \mathcal{A})$, is the smallest integer such that for any $T \subseteq [r]$ with $|T| \geq N(\alpha; \mathcal{A})$, there exists an $i \in T$ satisfying $\alpha A_i \preceq \sum_{j \in T \setminus \{i\}} A_j$.
    
    If no such integer exists, we define $N(\alpha; \mathcal{A}) = r + 1$.
}
\end{definition}

For a given strength $\alpha$, the parameter $N(\alpha;\mathcal{A})$ identifies the size of a critical subset of matrices. Specifically, it points to a subset $T$ of size $N(\alpha;\mathcal{A}) - 1$ where the sum $\sum_{i \in T} A_i$ cannot be sparsified further. This is because each matrix $A_i$ in the set is \emph{essential} to the sum, meaning it cannot be approximated by the others, as shown by the inequality $\alpha A_i \not\preceq \sum_{j \in T \setminus \{i\}} A_j$.

In case of graphs, when $A_i$s are edge Laplacians, one can think of $N(\alpha;\cA)$ as the minimum $k$ such that in every $k$ edge graph, every edge is ``robustly redundant''. Intuitively, the subgraph formed by removing any edge $e$ contributes at least an $\alpha$ fraction of the contribution of $e$ to any cut. For graphs with $n$ vertices, it can be shown that $N(\alpha;\cL_{\text{graphs}}) \geq n$ for all $\alpha > 0$. In fact, for $\alpha \in (0, 1/(n-1)]$, $N(\alpha;\cL_{\text{graphs}})$ is exactly $n$. More generally, the parameter is bounded by $(1+\alpha)n$ (see \cref{lem:graphconnparam}).

We now use the connectivity parameter to define a related \emph{connectivity threshold} for a collection of PSD matrices. One can think of our definition as a natural generalization of the threshold number of edges required for a graph to be connected. 

\begin{definition}[Connectivity Threshold]
Let $\alpha_{1/2} = (\sqrt{7/3} - 1)/2$ \footnote{This is just the number such that $\alpha_{1/2}(1 + \alpha_{1/2}) = (1 - 1/2)/(1 + 1/2) = 1/3$.} and $\cA = \{A_1, \ldots, A_r\}$ be a collection of PSD matrices. We define the \emph{connectivity threshold} of $\cA$ to be $N^*(\cA):= N(\alpha_{1/2}; \cA)$. 
\end{definition}

Our main theorem shows that connectivity threshold is an \emph{instance-specific} characterization of the size of a sparsifier up to polylogarithmic factors.

\begin{theorem}\label{thm:mainthmintro}
    Let $\cA = \{A_1,\ldots,A_r\} \subset \R^{n \times n}$ be a collection of PSD matrices. Then, for any $T \subseteq [n]$ and $\eps \leq 1/2$, there exists a sparse set of weights $\mu: T \rightarrow [0,\infty)$  such that 
    $$ \sum_{i \in T} \mu(i) A_i \approx_\eps \sum_{i \in T} A_i,$$

and 
\[|\supp(\mu)| \leq O(1) \left(\min_{\alpha \in (0,1]} \frac{N(\alpha;\cA)}{\alpha}\right) \cdot \left(\frac{(\log n)(\log r)}{\eps^2}\right)\mper\]

In particular, we have $|\supp(\mu)| \leq O(\eps^{-2}N^*(\cA)(\log n)(\log r))$. In addition, there is a randomized polynomial-time algorithm to find $\mu$. 

Complementarily, there exists $T \subseteq [n]$ of size $N^*(\cA) - 1$ such that $\cA_T$ cannot be $(1/2)$-sparsified any further. 
\end{theorem}

Thus $N^*(\cA)$ characterizes sparsifiability of sums of matrices from $\cA$ up to a $\eps^{-2}\log n\log r$ factor. We also state a more general version that gives a precise trade-off in the lower bound for sparsification for each $\eps > 0$ in \cref{thm:mainthmlowerbound,thm:mainthmupperbound}. 

\paragraph{Application to Cayley Sparsifiers for Cayley Graphs} We now return to our example of Cayley graph sparsification from before and show a $O(\log^4 |\cG|)$ size Cayley sparsifier for \emph{every} (possibly non-abelian) group. 

\begin{theorem}[See \cref{thm:cayleysparsification} for full version]\label{thm:cayleyintro}
    For any group $\cG$, and any symmetric set of generators $S$, there exists a symmetric (weighted) set of generators $T$ of size at most $O(\eps^{-2}\log^4 |\cG|)$ such that 
    \[L_{\Cay(\cG,T)} \approx_\eps L_{\Cay(\cG,S)}\mcom\]

    where $L_{\Cay(\cG,T)}, L_{\Cay(\cG,S)}$ denote the Laplacians of the Cayley graphs with generators $T, S$ respectively. 
\end{theorem}

Previously, such a sparsification result was only known for Cayley graphs over $\cG=\F_2^n$~\cite{KhannaPS24} via a beautiful connection to \emph{code sparsification} with applications in coding theory. \cite{KhannaPS24} showed that there always exists a sparsifier of size $O(\eps^{-2}n \polylog(n))$ for Cayley graphs over $\F_2^n$. Using our general result for arbitrary groups gives a bound of $O(\eps^{-2}n^4)$ for this special case. 

We prove our result on Cayley graph sparsification by applying our general framework for PSD matrix sparsification developed above that reduces the task to bounding the connectivity parameter:   

\begin{lemma}[See \cref{lem:cayleycyclethreshold} for full version]\label{lem:cayleyconnintro}
Let $\cG$ be a group with $N$ elements, and $\cL(\cG) = \{L_g: g \in \cG\}$ be a set of PSD matrices with $L_g$ being the Laplacian of the Cayley graph  generated by $\{g, g^{-1}\}$. Then, for $\alpha = \Omega(1/\log N)$, $N(\alpha; \cL(\cG)) = O(\log N)$.     
\end{lemma}

For every \emph{abelian} group $\cG$, we prove that $N^*(\cG) = \Omega(\log |\cG|)$ (see \cref{lem:abelianlb}). One can view this as a generalization of the well-known fact that one needs at least $\Omega(\log |\cG|)$ generators to get an expanding Cayley graph over any abelian group. We also prove (see \cref{lem:nonabelianlb}) that $N^*(\cG) \geq \Omega(\sqrt{\log |\cG|})$ for \emph{any} group $\cG$. As an immediate corollary, we conclude that for every group $\cG$, there is a Cayley graph on $\cG$ that does not admit a spectral sparsifier with $o(\sqrt{\log |\cG|})$ generators.

This may come as a surprise. The theory of spectral sparsification of graphs (every graph admits a $O(n)$ sparsifier) can be seen as a lossless generalization of the existence of $O(n)$ edge expander graphs. In an analogous manner, we know that there are $O(n)$-edge expanding Cayley graphs \cite{Margulis73, LubotzkyPS88, Morgenstern94} (indeed, these were the first explicit constructions of expander graphs in general!) on quasi-random groups \cite{Gowers08} (informally speaking, ``sufficiently non-abelian'') such as $\operatorname{SL}_2(\F_p)$. \cref{lem:nonabelianlb} shows that even for such groups, there are Cayley graphs that do not admit a constant-size sparsifier.

Note that for graphs we have $N^*(\cL_{\text{graphs}}) = \Theta(n)$ and thus can capture spectral sparsification of graphs up to a $O(\log^2 n)$ factor loss. We prove our main theorem for Cayley graphs, \cref{thm:cayleyintro} by using the above general bound along with using the group structure to bound the connectivity parameter for Cayley graphs:

\paragraph{Connections to Maximum Size of a Minimal Generating Set} The above theorem can be seen as a robust version of the following standard fact from group theory. For any group $\cG$, call a set of generators $S$ a \emph{minimal} generating set if dropping any element from $S$ does not generate the group $\cG$. In graph-theoretic terms, this means that the corresponding graph $\Cay(\cG,S)$ is connected, whereas dropping any element from $S$ leads to a Cayley graph that is disconnected. Let $m(\cG)$ denote the maximum size of a minimal generating set for $\cG$. Clearly, if we start with a minimal generating set, then the corresponding Cayley graph cannot be further sparsified by dropping any generator. This is sufficient to show that $N(\alpha; \cL(\cG))$ is at least $m(\cG)$. 

An elementary fact in group theory is that for any group $\cG$, $m(\cG) < \log_2 |\cG| + 1$ \cite{Gey12}. \cref{lem:cayleyconnintro} can be seen as a much more robust version of this theorem as we get not just connectivity but robust connectivity as quantified by the PSD ordering in the definition of $N(\alpha;\cL(\cG))$. 

Finally, it is easy to show that the connectivity threshold for the system of PSD matrices associated with a Cayley graph is at least $m(\cG)$. We also show that for all abelian groups $\cG$, $N^*(\cG) = \Omega(\log |\cG|)$ (see \cref{lem:abelianlb}). One can view this as a generalization of the well-known fact that one needs at least $\Omega(\log |\cG|)$ generators to get an expanding Cayley graph over any abelian group. 

\parhead{Implications for Code Sparsification}
Our theorems on PSD sparsification also yield results about code sparsification. Before we explain this connection, we recall the definition of code sparsification.
\begin{definition}[Code Sparsification]
    Let $C\subseteq\{0, 1\}^r$ be a (possibly non-linear) code. A map $w:[r]\to\R_{\geq 0}$ is called an $\eps$-sparsifier of $C$ if for every $c\in C$, $(1 - \eps)\wt(c)\leq\sum_{i\in[r]}w_ic_i\leq(1 + \eps)\wt(c)$, where $\wt(c) = \sum_{i\in[r]}c_i$ is the Hamming weight of $c$. 
\end{definition}
\cite{KhannaPS24, KhannaPS25} show that if $C$ is a \emph{linear} code,\footnote{i.e.\ $C\subseteq\F_2^r$ is a subspace} then in randomized $\poly(r, \eps^{-1})$ time, one can compute an $\eps$-sparsifier $w:[r]\to\R_{\geq 0}$ of the code $C$ with $|\supp(w)|\leq \widetilde{O}(\eps^{-2}\dim(C))$, where $\widetilde{O}$ suppresses factors in $\polylog(\dim(C))$. This result was extended by \cite{BrakensiekG25}, who showed that \emph{any} (possibly non-linear) code $C\subseteq\{0, 1\}^r$ admits an $\eps$-sparsifier of size $\leq O(\eps^{-2}\NRD(C)(\log r)^6)$, where $\NRD(C)$ is the \emph{non-redundancy} of the code $C$, which we define below. This result of \cite{BrakensiekG25} extends the results of \cite{KhannaPS24, KhannaPS25} since $\NRD(C) = \dim(C)$ if $C$ is linear.
\begin{definition}[Non-Redundancy of a Code]
    Given a code $C\subseteq\{0, 1\}^r$, $\NRD(C)$ is defined to be the largest integer $k$ such that there exists an index set $I\subseteq[r]$ with $|I| = k$ and codewords $c^{(1)}, \ldots, c^{(k)}\in C$ such that $c^{(j)}_i = \1(i = j)$ for all $i\in I, j\in[k]$. 
\end{definition}
Towards recovering the results of \cite{KhannaPS24, KhannaPS25, BrakensiekG25}, we also construct a family of PSD matrices $\cC$ given any code $C\subseteq\{0, 1\}^r$: Indeed, define $\cC:= \{C_1, \ldots, C_r\}$, where for any $i\in[r]$, $C_i\in\R^{|C|\times |C|}$ is a diagonal matrix whose entries are indexed by elements of $C$, where for any $c\in C$, $C_i(c, c):= c_i\in\{0, 1\}$. Note that the matrices $C_1, \ldots, C_r$ are PSD, and if we write $S:= \sum_{i\in[r]}C_i$, then $S$ is a diagonal matrix satisfying $S(c, c) = \wt(c)$. Consequently, the task of PSD sparsification of $\cC$ exactly corresponds to the task of sparsifying the code $C$! Furthermore, our notion of connectivity threshold also coincides with the notion of non-redundancy, as demonstrated below:
\begin{lemma}
\label{lem:NRDConnection}
    For any $C\subseteq\{0, 1\}^r$, $N(\alpha;\cC) = \NRD(C) + 1$ for all $\alpha\in(0, 1]$.
\end{lemma}
\begin{proof}
    We will show that $N(\alpha;\cC)\geq\NRD(C) + 1$ for all $\alpha > 0$, and $N(1;\cC) \leq\NRD(C) + 1$. Since $\alpha\mapsto N(\alpha;\cC)$ is a non-decreasing function, this concludes the result.

    \parhead{$N(\alpha;\cC)\geq\NRD(C) + 1$ for all $\alpha > 0$} Write $\NRD(C) = k$, let $I\subseteq[r]$ be an index set of size $k$, let codewords $c^{(1)}, \ldots, c^{(k)}\in C$ be such that $c^{(j)}_i = \1(i = j)$ for all $i\in I, j\in[k]$, and consider $\cC':= \{C_i\}_{i\in I}\subseteq\cC$. Note that for any $\alpha > 0, i\in I$, we have $\alpha C_i\not\preceq\sum_{i'\in I\setminus\{i\}}C_{i'}$ since $C_i(c^{(i)}, c^{(i)}) = 1$, while $C_{i'}(c^{(i)}, c^{(i)}) = 0$ for all $i'\in I\setminus\{i\}$. Thus $N(\alpha;\cC)\geq\NRD(C) + 1$, as desired.
    
    \parhead{$N(1;\cC)\leq\NRD(C) + 1$} Write $N(1;\cC) = k + 1$. By definition of $N(\ast)$, there must exist $I\subseteq[r]$ of size $k$ such that for all $i\in I$, $C_i\not\preceq\sum_{i'\in I\setminus\{i\}}C_{i'}$. Consequently, for every $i\in I$, there exists $x^{(i)}\in\R^{|C|}$ such that $x^{(i)\top}C_ix^{(i)} > \sum_{i'\in I\setminus\{i\}}x^{(i)\top}C_{i'}x^{(i)}\implies\sum_{c\in C}C_i(c, c)(x^{(i)}_c)^2 > \sum_{i'\in I\setminus\{i\}}C_{i'}(c, c)(x^{(i)}_c)^2\implies\sum_{c\in C}(x^{(i)}_c)^2(C_i(c, c) - \sum_{i'\in I\setminus\{i\}}C_{i'}(c, c)) > 0$. Since $(x^{(i)}_c)^2\geq 0$ for all $c\in C$, there must exist some $c^{(i)}\in C$ such that $C_i(c^{(i)}, c^{(i)}) > \sum_{i'\in I\setminus\{i\}}C_{i'}(c^{(i)}, c^{(i)})\geq 0$. Since $C_i(c^{(i)}, c^{(i)}) > 0$, we have $C_i(c^{(i)}, c^{(i)}) = 1$. Since $\sum_{i'\in I\setminus\{i\}}C_{i'}(c^{(i)}, c^{(i)}) < 1$, we have $C_{i'}(c^{(i)}, c^{(i)}) = 0$ for all $i'\in I\setminus\{i\}$, i.e.\ $C_{i'}(c^{(i)}, c^{(i)}) = c^{(i)}_{i'} = \1(i = i')$ for all $i, i'\in I$, and we're done.
\end{proof}
As an immediate consequence of \cref{lem:NRDConnection,thm:mainthmintro} we have the following result:
\begin{theorem}
\label{thm:BGConn}
    For any code $C\subseteq\{0, 1\}^r$ in randomized $\poly(r, |C|, \eps^{-1})$ time, one can compute an $\eps$-sparsifier $w:[r]\to\R_{\geq 0}$ with $|\supp(w)|\leq O(\eps^{-2}\NRD(C)\cdot(\log|C|)(\log r))$. 
\end{theorem}
\cref{thm:BGConn} should be seen as analogous to \cite[Theorem~3.1]{BrakensiekG25}, who do something essentially similar to leverage score sparsification to obtain a sparsifier of size $O(\eps^{-2}\NRD(C)\cdot(\log|C|)(\log^3 r))$. In light of the result above, the key innovation in \cite{BrakensiekG25} is to remove the $\log|C|$ term arising from a naive union bound over $|C|$ different ``directions''. Alternatively one might also interpret the $\log|C|$ term as arising from a direct invocation of the matrix concentration inequality (see \cref{lem:psdsparsification}).

\parhead{Concurrent Work} In a concurrent and independent work, Hsieh, Lee, Mohanty, Putterman, and Zhang \cite{HsiehLMPZ25} obtained a sparsification result for Cayley graphs as in \cref{thm:cayleyintro} with similar parameters.
\subsection{Proof Outline}
We next explain the main ideas in the proof of \cref{thm:mainthmintro}. The crux is the upper bound's proof; the lower bound follows from the definition of the connectivity threshold. For this discussion, consider a fixed collection of PSD matrices $\cA = \{A_1,\ldots,A_r\}$, a fixed $\alpha \in (0,1]$, and $\eps > 0$, as in the theorem statement.

To obtain sparsifiers, we use the technique of \emph{leverage score sparsification} or importance sampling \cite{SpielmanS11}. We first recall this standard argument. Given a subset $T \subseteq [r]$, we select a subset by keeping each element $i \in T$ with a certain probability $p_i$. For $i \in T$, let $X_i \in \{0,1\}$ be independent indicator random variables with $\Pr[X_i = 1] = p_i$. We approximate $A_T := \sum_{i \in T} A_i$ with the corresponding weighted sum $\hat{A}_T := \sum_{i \in T} (1/p_i) X_i A_i$. Clearly, $\E[\hat{A}_T] = A_T$. The proper probabilities to use are $p_i = \min\{1, \|\tilde{A}_i\|_2\}$, where $\tilde{A}_i = (A_T)^{-1/2} A_i (A_T)^{-1/2}$. This leads to the following statement (see \cref{lem:psdsparsification}): We can find $\mu:T \rightarrow \mathbb{R}_{\geq 0}$ such that $\sum_{i \in T} \mu_i A_i \approx_\eps \sum_{i \in T} A_i$ and
$$|\supp(\mu)| \leq O(\eps^{-2}\log n) \cdot \left(\sum_{i \in T} \|\tilde{A}_i\|_2\right).$$

Thus, it remains to bound $\sum_{i \in T} \|\tilde{A}_i\|_2$. We show this sum is at most $O(\log r) \cdot N(\alpha;\cA)/\alpha$. To do this, we fix a threshold $\eta$ and let $S_\eta := \{i \in T : \|\tilde{A}_i\|_2 \geq \eta\}$. We argue that for any $\eta$, $|S_\eta| = O(N(\alpha;\cA)/\alpha \eta)$. This bound allows us to bound $\sum_i \|\tilde{A}_i\|_2$ by simple integration, which introduces a $O(\log r)$ factor.

Intuitively, we bound $|S_\eta|$ by arguing that if $|S_\eta|$ is larger than $4 N(\alpha;\cA)/\alpha \eta$, a random subset will violate the definition of the connectivity parameter $N(\alpha;\cA)$. Indeed, with positive probability there exists a subset $S'_\eta\subset S_\eta$ with $|S'_\eta| = \Omega(\alpha\eta|S_\eta|)$ such that $\alpha A_i \not\preceq \sum_{j \in S'_\eta: j \neq i} A_j$ for all $i\in S'_\eta$. Consequently, $|S'_\eta| < N(\alpha;\cA)$ by the very definition of $N(\cdot;\cA)$, and thus we have $\Omega(\alpha\eta|S_\eta|) < N(\alpha;\cA)$, as desired.

\paragraph{Cayley Graph Sparsification } We next sketch the main idea for the proof of \cref{thm:cayleyintro}. Fix a group $\cG$. As mentioned, the proof follows from \cref{thm:mainthmintro} by showing that $N(\alpha;\cL(\cG)) = O(\log N)$ for $\alpha = \Omega(1/\log N)$ (where $N = |\cG|$), as in \cref{lem:cayleyconnintro}. Consider an arbitrary subset $T \subseteq \cG$ of size $\lceil \log_2 N \rceil + 1$. We show this set is not \emph{minimal}, meaning there exists an element $s \in T$ such that for $\alpha = \Omega(1/\log_2 N)$,
$$\alpha L_s \preceq \sum_{g \in T\setminus\{s\}} L_g.$$

To prove this, we use the fact that since $|T| > \log_2 N$, a non-trivial relation must exist among the elements of $T$. Specifically, some element $s \in T$ can be written as a product of others: $s = t_1 t_2 \cdots t_\ell$, where each $t_i \in T\setminus \{s\}$. Using this algebraic relation and a careful analysis of the corresponding Laplacian matrices, we can show the above inequality holds.

\subsection{Discussion}
Our work raises a number of further interesting questions. Foremost among these is whether it is possible to get an analog of the results that \cite{SpielmanS11} (i.e., save one logarithmic factor) or \cite{BatsonSS14} (no logarithmic factors at all) get for general PSD matrices:

\begin{question}\label{qmain}
    Given any collection $\cA = \{A_1,\ldots,A_r\} \subset \R^{n \times n}$, can we always sparsify any sum of elements of matrices from $\cA$ to only have $O(N^*(A))$ elements for a constant $\eps$ (say $.9$)? 
\end{question}

A positive result above would be very interesting and could also have strong implications in group theory. 

In the setting of Cayley graphs for arbitrary groups $\cG$, we establish an upper bound of $O(\log^2 |\cG|)$ on the quantity $\min_\alpha (N(\alpha;\cA)/\alpha)$. This contrasts with the lower bounds we show. For general groups, the only established lower bound is $\max\{m(\cG), \Omega(\sqrt{\log |\cG|})\}$, where $m(\cG)$ is the maximum size of a minimal generating set. For abelian groups, we show a tighter bound of $\Omega(\log |\cG|)$.

It remains plausible that the connectivity threshold is substantially smaller for groups with strong mixing properties, such as quasi-random groups; especially, if one looks at a subset of generators instead of all generators. Indeed, a positive resolution to  \cref{qmain} would likely have implications for identifying expanding generator sets in such groups (e.g.,  similar to the results of Bourgain and Gamburd \cite{BourgainGamburd}). 

\section{Preliminaries}
We start off by recalling some basic facts and notions from linear algebra.
\subsection{Linear Algebra}
Throughout the paper all vector spaces should be assumed to be real and finite-dimensional.

\parhead{Hermitian and Unitary Matrices} Let $V$ be a vector space equipped with an inner product $\langle x, y\rangle = x^*y$, and let $M:V\to V$ be a linear map. Denote by $M^*$ the adjoint of $M$. $M$ is called Hermitian if $M = M^*$. An invertible linear map $U:V\to V$ is called unitary if $U^{-1} = U^*$. Unitary maps are isometries, i.e.\ $\|Uv\|_2 = \|v\|_2$ for all $v$. Also, all eigenvalues of $U$ are unit complex numbers, i.e.\ if $Uv = \lambda v$ for some $\lambda\in\C$, then $|\lambda| = 1$. 

\parhead{PSD Matrices} A Hermitian map $M:V\to V$ is called \emph{Positive Semi-Definite} (PSD) if $\langle v, Mv\rangle\geq 0$ for all $v$. For any two linear maps $M_1, M_2:V\to V$, we write $M_1\preceq M_2$ if $M_2 - M_1$ is Hermitian PSD. If $M_1\preceq M_2$, then for any linear map $A:V\to V$, we have $M_1 + A\preceq M_2 + A$ and $A^*M_1A\preceq A^*M_2A$.

From now on, we'll use the terms ``linear map'' and ``matrix'' interchangeably.

\parhead{Symmetrization of Matrices} For any square matrix $M$, define $H_M:= (M + M^*)/2$. Note that $H_M$ is Hermitian. If $U$ is unitary, then $H_U$ is a Hermitian matrix with eigenvalues in $[-1, 1]$, and $\Id - H_U$ is a PSD matrix.  

\parhead{Norms} For any vector $v$, $\|v\|$ should be taken to mean the $\ell_2$-norm $\|v\|_2$ unless stated otherwise. Similarly, for any matrix $M$, $\|M\|$ should be taken to mean $\|M\|_2:= \sup_{\|v\|_2 = 1}\|Mv\|_2$ unless stated otherwise.

Finally, for any two Hermitian PSD matrices $A, B$, write $A\approx_\eps B$ to mean $(1 - \eps)B\preceq A\preceq (1 + \eps)B$.

\subsection{Laplacians}
Let $V$ be a vertex set, and let $i, j\in V$. We define the \emph{Laplacian} $L_e\in\R^{V\times V}$ of the edge $e = \{i, j\}$ as 
\[L_e:= \begin{bmatrix}
0      & \cdots &        &        & 0 \\
\vdots & 1      & \cdots & -1     & \vdots \\
       & \vdots & \ddots & \vdots &        \\
       & -1     & \cdots & 1      & \vdots \\
0      & \cdots &        & \cdots & 0
\end{bmatrix}\mcom\]
where the $i^{\mathrm{th}}, j^{\mathrm{th}}$ diagonals have ones, and the $(i, j)^{\mathrm{th}}, (j, i)^{\mathrm{th}}$ positions have $-1$ in them. Note that $L_e$ is rank $1$, and for any $x\in\R^V$, $x^*L_ex = (x_i - x_j)^2$, which immediately demonstrates that $L_e$ is PSD. 

For any undirected graph $G = G(V, E)$ (not necessarily simple), define the Laplacian $L_G\in\R^{V\times V}$ of $G$ as $L_G:= \sum_{e\in E}L_e$. If $G$ is weighted, i.e.\ we have weights $w:E\to\R_{\geq 0}$, then the Laplacian is defined as $\sum_{e\in E}w(e)L_e$.

For any undirected graph (not necessarily simple, possibly weighted), the Laplacian is a Hermitian PSD matrix such that $L\1 = 0$, where $\1\in\R^V$ is the all-ones vector. Furthermore, the multiplicity of zero as an eigenvalue of $L$ is equal to the number of connected components of $G$.\footnote{In case $G$ is weighted, consider the edges in $\supp(w)$, ignore their weights, and the number of connected components of $G$ is defined as the number of connected components of the resulting unweighted graph}

We shall need the Matrix Bernstein inequality to analyze our sampling scheme:
\begin{fact}[Matrix Bernstein Inequality, \cite{Tro15}]
\label{fact:matrixbernstein}
    Let $Y_1, \ldots, Y_m\in\R^{n\times n}$ be independent mean $0$ random Hermitian matrices such that $\|Y_i\|_2\leq R$ a.s. for all $i\in[m]$. Write $Y = \sum_{i\in[m]}Y_i$, and $\sigma^2:= \|\sum_{i\in[m]}\E[Y_i^2]\|_2$. Then 
    \[\Pr(\|Y\|\geq t)\leq 2n\cdot\exp\left(-\frac{t^2}{\sigma^2 + Rt/3}\right)\mper\]
\end{fact}
\parhead{Miscellaneous} For any $x\geq 1$, define $\log_+(x):= \max\{1, \log x\}$. Throughout what follows we shall abuse notation and write $O(\log t)$ instead of $O(\log_+(t))$, where $t$ is the pertinent parameter. Thus, when $t = 1$, $O(\log t)$ should be read as $O(1)$. 

Finally, $\log$ will mean $\log_e$ unless stated otherwise. 

For any set $X$ and any map $\mu:X\to\R$, define $\supp(\mu):= \{x\in X:\mu(x)\neq 0\}$. Also, for any $x\in X$ we will use the terms $\mu(x)$ and $\mu_x$ interchangeably.
\subsection{Cayley Graphs}
Let $\cG$ be any finite group (possibly non-abelian). Write $N:= |\cG|$. Let $S\subseteq \cG$ be a symmetric subset, i.e.\ $S = S^{-1}$. 
\begin{definition}[Cayley graphs]
    Given a group $\cG$ and a symmetric subset $S\subseteq \cG$, the Cayley graph $\Cay(\cG, S)$ is a graph with the vertex set $\cG$ and $g_1, g_2\in \cG$ are adjacent if $g_1^{-1}g_2\in S$.
\end{definition}
The elements of $S$ are also known as \emph{generators} of the Cayley graph.

Note that $\Cay(\cG, S)$ is undirected since $S$ is symmetric: Indeed, if $S\ni s = g_1^{-1}g_2$, then $g_2^{-1}g_1 = s^{-1}\in S$. Also note that $\Cay(\cG, S)$ is connected if and only if $S$ generates $\cG$.

\parhead{Weighted Cayley Graphs and Symmetric Functions} If we have a symmetric weight function $w:S\to\R_{\geq 0}$, where symmetry simply means that $w(s) = w(s^{-1})$ for all $s\in S$, then $\Cay(\cG, S)$ naturally becomes a weighted graph $\Cay(\cG, S, w)$, where edges induced by the generators $s$ or $s^{-1}$ are given weight $w(s)$.

\parhead{The Right Regular Representation} Consider the vector space $\R \cG:= \{\sum_{g\in \cG}\alpha_gg: \alpha_g\in\R\text{ for all }g\in \cG\}$. Here the sum $\sum_{g\in \cG}\alpha_gg$ should be treated as a formal sum, and $\R \cG$ is equipped with the obvious notions of addition ($\sum_g \alpha_gg + \sum_g\beta_gg:= \sum_g(\alpha_g + \beta_g)g$), scaling ($\lambda\cdot\sum_g\alpha_gg := \sum_g(\lambda\alpha_g)g$) and inner product ($\langle\sum_g\alpha_gg, \sum_g\beta_gg\rangle:= \sum_g\alpha_g\beta_g$). 

Now, for any $g\in \cG$, define the right regular representation $\scrR(g):\R \cG\to\R \cG$,\footnote{$\scrR(\cdot)$ is more commonly defined as a map from $\C \cG\to\C \cG$ but we restrict to reals in this work.} where 
\[\scrR(g)\left(\sum_{h\in \cG}\alpha_hh\right):= \sum_{h\in \cG}\alpha_hhg = \sum_{h\in \cG}\alpha_{hg^{-1}}h\mper\]
Note that $\scrR(g)$ is a linear map, and it is unitary (since its action permutes the basis $\{\1_g\}_{g\in \cG}$ of $\R \cG$). Furthermore, it is a representation since $\scrR(g_1g_2) = \scrR(g_1)\scrR(g_2)$ for all $g_1, g_2\in \cG$. In particular, since $\Id = \scrR(\id_{\cG})$, we have $\Id = \scrR(\id_{\cG}) = \scrR(x\cdot x^{-1}) = \scrR(x)\cdot\scrR(x^{-1})\implies\scrR(x^{-1}) = \scrR(x)^{-1} = \scrR(x)^*$, where $\scrR(x)^{-1} = \scrR(x)^*$ since $\scrR(x)$ is unitary. 

\parhead{Laplacians of Cayley graphs} Now, note that the adjacency operator of the edges induced by the generator $s$ is $\scrR(s)$ (where we view the adjacency operator as acting on $\R\cG$), and thus the adjacency operator of $\Cay(\cG, S)$ is $\sum_{s\in S}\scrR(s) = \sum_{s\in S}(\scrR(s) + \scrR(s^{-1}))/2 = \sum_{s\in S}(\scrR(s) + \scrR(s)^*)/2$, where the first equality follows since $S = S^{-1}$.

Consequently, the Laplacian of $\Cay(\cG, S)$ is $|S|\cdot\Id - \sum_{s\in S}(\scrR(s) + \scrR(s)^*)/2 = \sum_{s\in S}(\Id - (\scrR(s) + \scrR(s)^*)/2) = \sum_{s\in S}L_s$, where $L_s:= \Id - (\scrR(s) + \scrR(s)^*)/2 = \Id - H_{\scrR(s)}$. Note that $L_s + L_{s^{-1}}   = 2L_s$ is the Laplacian of the edges induced by $\{s, s^{-1}\}$.

If we have a symmetric weight function $w:S\to\R_{\geq 0}$, then the Laplacian of $\Cay(\cG, S, w)$ is $\sum_{s\in S}w(s)L_s$.

Now, note that $L_s$ is Hermitian. Also, since $\scrR(s)$ is unitary, $H_{\scrR(s)}$ is a Hermitian matrix with eigenvalues in $[-1, 1]$, and $L_s = \Id - H_{\scrR(s)}$ is a Hermitian PSD matrix.
\subsubsection{Eigenvalues of Abelian Cayley Graphs}
Let $\cG = (\cG, +, 0)$ be an abelian group of size $N$. We first state the \emph{fundamental theorem of abelian groups}, which states that all finite abelian groups are isomorphic to a direct sum of cyclic groups $\Z_k:= \Z/k\Z$:
\begin{fact}[Fundamental Theorem of Abelian Groups]
\label{fact:fundabel}
    Any finite abelian group $\cG$ is isomorphic to a direct sum $\bigoplus_{i = 1}^r\Z_{k_i}$, where $k_1, \ldots, k_r\geq 2$ are prime powers, and $|\cG| = k_1\cdots k_r$.
\end{fact}
Using the above fact, and the eigenvectors of cyclic groups, we can describe the eigenvectors and eigenvalues of the matrices $L_s$ for $s\in\cG$:
\begin{fact}
\label{fact:abelianeig}
    Let $\cG$ be an abelian group of size $N$ isomorphic to $\bigoplus_{i = 1}^r\Z_{k_i}$. For any $(g_1, \ldots, g_r)\in\Z^r$ define $\chi_{g_1, \ldots, g_r}\in\R^{\cG}$ as 
    \[\chi_{g_1, \ldots, g_r}(h):= \cos\left(2\pi\sum_{i = 1}^r\frac{g_ih_i}{k_i}\right)\]
    where $h = (h_1, \ldots, h_r)\in \bigoplus_{i = 1}^r\Z_{k_i}$ is the representation of $h\in\cG$ in the direct sum. Then for any $s = (s_1, \ldots, s_r)\in\cG$, we have 
    \[L_s\chi_{g_1, \ldots, g_r} = 2\sin^2\left(\pi\sum_{i = 1}^r\frac{g_is_i}{k_i}\right)\chi_{g_1, \ldots, g_r}\mcom\]
    i.e.\ $\chi_{g_1, \ldots, g_r}$ is an eigenvector of $L_s = \Id - (\scrR(s) + \scrR(-s))/2$ with eigenvalue $2\sin^2\left(\pi\sum_{i = 1}^r\frac{g_is_i}{k_i}\right)$. 
\end{fact}

\subsection{Schreier Graphs}
Given that we are able to obtain $\polylog(N)$ sized sparsifiers for Cayley graphs, one might wonder if it is possible to obtain similar results for Schreier graphs. 

We show (\cref{lowerboundsparsgraph,schreiersparlowerbound}) that the answer is strongly negative, essentially because the class of Schreier graphs encompasses \emph{all} regular graphs of even degree. 

We now set up some notation to define Schreier graphs, for which we first define group actions:
\begin{definition}[Group Actions]
    Let $\cG$ be a group, and let $X$ be a set. We say that $\cG$ \emph{acts on} $X$, denoted as $\cG\curvearrowright X$, if we have a map $\bm{\cdot}: X\times \cG\to X$ which satisfies the following axioms for all $x\in X, g, h\in G$:
    \begin{enumerate}[(1)]
        \item $x\bm{\cdot}\id_{\cG} = x$, 
        \item $(x\bm{\cdot} g)\bm{\cdot} h = x\bm{\cdot} (gh)$.
    \end{enumerate}
\end{definition}
A group action $\cG\curvearrowright X$ is called \emph{transitive} if for any $x\neq y\in X$ there exists $g\in \cG$ such that $x\cdot g = y$.

We can now define Schreier graphs:
\begin{definition}[Schreier Graphs]
    Let $\cG$ be a group, and let $X$ be a set such that $\cG$ acts on $X$. Let $S\subseteq \cG$ be a symmetric subset of $\cG$. Then the Schreier graph $\Sch(\cG, S, X)$ is a graph with vertex set $X$, and edge set $\{(x, x\cdot s):s\in S\}$. 
\end{definition}
The elements of $S$ are also known as \emph{generators} of the Schreier graph.

Note that since $S = S^{-1}$, if $(x, x\cdot s)$ is an edge, then $(x\cdot s, (x\cdot s)\cdot s^{-1}) = (x\cdot s, x)$ is also an edge, and thus $\Sch(\cG, S, X)$ is an undirected graph. Given a symmetric weight function $w:S\to\R_{\geq 0}$, $\Sch(\cG, S, X)$ naturally becomes a weighted graph.

\parhead{Schreier Coset Graphs} An example of Schreier graphs are the so-called \emph{Schreier coset graphs}: Let $\cG$ be a group, let $\cH$ be a subgroup of $\cG$, and let $\cG\backslash \cH:= \{\cH g:g\in \cG\}$ be the set of right cosets of $\cH$. Then $\cG$ acts on $\cG\backslash \cH$ as $(\cH g)\cdot g':= \cH gg'$, and for any symmetric $S\subseteq \cG$, $\Sch(\cG, S, \cG\backslash \cH)$ is called a Schreier coset graph. Note that if $H = \{\id_{\cG}\}$, then $\Sch(\cG, S, \cG\backslash \{\id_{\cG}\})\cong\Cay(\cG, S)$.

\parhead{Expressivity of Schreier Graphs} Despite the ``algebraic'' definition of Schreier graphs, a surprising but straightforward to prove result due to \cite{Gro77} says that \emph{every} regular graph of even degree (without loops of degree $1$) is isomorphic to some Schreier graph. We prove a result similar to this for our purposes.

Recall that an undirected graph is called simple if any edge occurs at most once in the graph. Also, a perfect matching is a graph in which every vertex has degree exactly $1$.
\begin{lemma}[Similar to Theorem 2 of \cite{Gro77}]
\torestate{
\label{matchingschreier}
    Let $G = (V, E)$ be a connected undirected graph, not necessarily simple. Suppose $E = \bigsqcup_{i = 1}^dM_i$, where $M_i$ is a perfect matching on $V$ for all $i\in[d]$. Then $G$ is isomorphic to a Schreier graph $\Sch(\cG, S, V)$, where $s = s^{-1}$ for all $s\in S$, and for every $s\in S$ there is a unique $i\in [d]$ such that $M_i$ is the set of edges induced by $s$.
}
\end{lemma}
We prove this statement in \cref{matchingschreierproof}.
\section{Sparsifying sums of PSD matrices}
In this section, we present our main theorem for sparsifying sums of arbitrary positive semidefinite matrices. 

\subsection{Connectivity Thresholds for PSD Matrices}
In this section, we discuss our key conceptual contribution, \emph{connectivity parameter}, and its properties. 

\restatedefinition{def:connparam}
We may drop $\alpha$ or $\cA$ from the notation when it is clear from the context. 

For any $\alpha > 0$, if $N(\alpha;\cA) = k$, then by definition there must exist subsets $\cB\subset\cA$ of size $< k$ such that for all $B\in\cB$, $\alpha B\not\preceq\sum_{A\in\cB\setminus\{B\}}A$. Such subsets $\cB$ are called \emph{$\alpha$-minimal} subsets of $\cA$. 

As the name suggests, the connectivity parameter generalizes that every graph with $n$ vertices and $n$ edges has a non-trivial connected subgraph. Our main theorem can then be seen as a generalization of the fact that every graph admits a spectral sparsifier with the number of edges equal to the connectivity threshold up to a polylogarithmic (in fact even a constant) factor. We next discuss the example of graphs in more detail below.

\begin{lemma}[Connectivity Parameters for Graphs]
\label{lem:graphconnparam}
    Let $G = G(V, E)$ be a connected graph with $n$ vertices, and let $\cL:= \{L_e: e\in E(G)\}$ be the set of edge Laplacians of the edges of $G$. Then:
    \begin{enumerate}[(1)]
        \item $N(\alpha; \mathcal{L}) = n$ for all $0 < \alpha \leq \frac{1}{n - 1}$.
        \item $N(\eps; \mathcal{L})\leq \lceil(1 + \eps)n\rceil$ for any $\eps > 0$.
    \end{enumerate}
\end{lemma}
\begin{proof}
Firstly, we claim that $N(\alpha) \geq n$ for all $\alpha > 0$.\footnote{By definition, $N(0) = 1$.} To show this, we claim that if $T\subseteq E$ is a spanning tree (with $n - 1$ edges) in $G$, then $\{L_e:e\in T\}$ is $\alpha$-minimal for all $\alpha > 0$. 

To show this, we have to show that for any $e\in T$ and any $\alpha > 0$, $\alpha L_e\not\preceq\sum_{e'\in T\setminus\{e\}}L_{e'}$, or equivalently, for every $e\in T$ there exists a vector $v_e\in\R^V$ such that for all $\alpha > 0$, $\alpha v_e^*L_ev_e > \sum_{e'\in T\setminus\{e\}}v_e^*L_{e'}v_e$, i.e.\ $v_e^*L_ev_e > 0 = \sum_{e'\in T\setminus\{e\}}v_e^*L_{e'}v_e$.

Towards this, note that for any $e\in T$, deleting $e$ from $T$ disconnects $T$, i.e.\ for every $e\in T$ there exists a subset $S_e\subseteq V$ such that $e$ is the only edge in the cut $(S_e, V\setminus S_e)$ in $T$. Let $v_e\in\Fits^V$ be such that $v_e(x) = 1$ if $x\in S_e$, $-1$ otherwise. Then $v_e^*L_ev_e > 0$, while $\sum_{e'\in T\setminus\{e\}}v_e^*L_{e'}v_e = 0$, as desired.

Finally, for any subset $E'\subseteq E$ of size $\geq n$, we can find a cycle of length $k\leq n$ in $E'$. Write $e_i = \{v_i, v_{i + 1}\}$ for all $1\leq i\leq k$, where the indices are taken modulo $k$, and note that for any $x\in\R^V$, $\frac{1}{k - 1}x^*L_{e_1}x = \frac{1}{k - 1}(x_{v_1} - x_{v_2})^2\overset{\text{Cauchy-Schwarz}}{\leq}\sum_{i > 1}(x_{v_i} - x_{v_{i + 1}})^2 = \sum_{i > 1}x^*L_{e_i}x$, thus showing that $\frac{1}{k - 1}L_{e_1}\preceq\sum_{i > 1}L_{e_i}$, and thus $\frac{1}{n - 1}L_{e_1}\preceq\frac{1}{k - 1}L_{e_1}\preceq\sum_{i > 1}L_{e_i}\preceq\sum_{e'\in E'\setminus\{e_1\}}L_{e'}\implies \frac{1}{n - 1}L_{e_1}\preceq\sum_{e'\in E'\setminus\{e_1\}}L_{e'}$, as desired.

Now, let $E'\subseteq E$ be any subset of size $\geq (1 + \eps)n$, let $L$ be the Laplacian of $E'$, and write $W:= \ker(L)^\perp$. Note that $L\vert_W$ is invertible, and thus write $\tL_e:= L^{-1/2}L_eL^{-1/2}$ to be the normalized edge Laplacian on $W$. Note that $\sum_{e\in E'}\|\tL_e\| = \sum_{e\in E'}\Tr(\tL_e) = \Tr\left(\sum_{e\in E'}\tL_e\right) = \Tr(L^{-1/2}LL^{-1/2}) = \Tr(\Id_{W}) = \dim(W) \leq n$, where the first equality follows since $L_e$ is rank $1$ for all $e$. Consequently, if $e^*\in E'$ is such that $\|\tL_{e^*}\|$ is the smallest, then $\|\tL_{e^*}\| \leq \frac{n}{|E'|}\leq \frac{1}{1 + \eps}$, i.e.\ $\tL_{e^*}\preceq\frac{1}{1 + \eps}\Id\implies (1 + \eps)\tL_{e^*}\preceq\Id = \sum_{e\in E'}\tL_e\implies\eps\tL_{e^*}\preceq\sum_{e\in E'\setminus\{e^*\}}\tL_{e}\implies \eps L_{e^*}\preceq\sum_{e\in E'\setminus\{e^*\}}L_{e}$, as desired.
\end{proof}

Next, we discuss the example of Cayley graphs. Here, each $A_i$ is the Laplacian matrix of the graph with all edges corresponding to a single generator of the group. We observe that the connectivity parameter corresponds to a natural and well-studied quantity called the maximum size of any minimal generating set. This connection will be crucial to obtain our sparsification result for Cayley graphs.

\begin{lemma}[Connectivity Parameters for Cayley graphs]
\label{lem:cayleyconnparam}
    Let $\cG$ be a group of size $N$, and let $S\subseteq \cG$ be a symmetric set. Write $\cL:= \{L_s:s\in S\}$. Then $N(\frac{1}{2N^2}; \cL) \leq m(\cG) + 1$, where $m(\cG)$ is the maximum possible size of a minimal generating set of $\cG$.
\end{lemma}
\begin{proof}
    Let $X\subseteq S$ be any set of size $\geq m(\cG) + 1$, and let $\cH$ be the subgroup generated by $X$. Since $|X|\geq m(\cG) + 1$, $X$ can not be a minimal generating set of $\cH$, and thus there exists $x\in X$ such that $x = x_1\cdots x_k$, where $x_1, \ldots, x_k\in X\setminus\{x, \id_{\cG}\}$ are not necessarily distinct. Note that we can assume $k\leq n$, since $k$ is the distance (in the graph metric) of $x$ from $\id_{\cG}$ in $\Cay(\cG, \{x_1, \ldots, x_k\})$, which has to be $\leq |\cG| = N$. Also note that WLOG $x_i \neq x^{-1}$ for all $i\in[k]$, since otherwise $L_x = L_{x^{-1}} = L_{x_i}\preceq\sum_{y\in X\setminus\{x\}}L_y$, and we're done. 

    Now, write $L_x = \sum_{e\in E_x} L_{e}$, where $E_x$ is the set of edges induced by $\{x, x^{-1}\}$ in $\Cay(\cG, S)$, and for any $e\in E_x$, $L_{e}$ denotes the edge Laplacian of the edge $e$. Now, for any $g\in S, s\in S$, denote the edge $\{g, gs\}$ as $e_{g, s}$, and note that for all $g\in G$, the edges $e_{g, x}, e_{gx, x_k^{-1}}, e_{gxx_k^{-1}, x_{k - 1}^{-1}}, \ldots, e_{gxx_k^{-1}\cdots x_2^{-1}, x_1^{-1}}$ form a cycle, and thus, as in \cref{lem:graphconnparam}, we have 
    \begin{align}
    \label{eq:cycleeq}
        \frac{1}{k}L_{e_{g, x}}\preceq\sum_{i = 1}^k L_{g\prod_{j > i}x_j^{-1}, x_i^{-1}}\mcom
    \end{align}
    where $\prod_{j > i}x_j^{-1}$ stands for $x_k^{-1}\cdots x_{i + 1}^{-1}$ when $i < k$, and $\id_{\cG}$ when $i = k$. Adding \cref{eq:cycleeq} for all $g\in G$, we obtain that 
    \[\frac{1}{2k}L_x\preceq\sum_{i = 1}^k L_{x_k}\]
    on observing that for any $x\in S$, $L_x = \sum_{g\in G}L_{e_{g, x}}$ if order of $x$ is $ > 2$, and $L_x = \frac{1}{2}\sum_{g\in G}L_{e_{g, x}}$ if order of $x$ is $2$.\footnote{WLOG we assume $x\neq\id_{\cG}$ since $L_{\id_{\cG}} = 0$, and we're trivially done then} On the other hand, the multiplicity of any $y\in X\setminus\{x\}$ in the multiset $\{x_1, \ldots, x_k\}$ is $\leq k$, and thus $\frac{1}{k}\sum_{i = 1}^k L_{x_k}\preceq\sum_{y\in X\setminus\{x\}}L_y$, and thus 
    \[\frac{1}{2N^2}L_{x}\preceq\frac{1}{2k^2}L_{x}\preceq\sum_{y\in X\setminus\{x\}}L_y\mper\]
    as desired.
\end{proof}
Our main theorems \cref{thm:mainthmlowerbound,thm:mainthmupperbound} show that the \emph{connectivity threshold}, defined below, characterizes the size of sparsifier up to a polylogarithmic factor for any collection of PSD matrices. 

\begin{definition}[Connectivity Threshold]
    Let $\eps\in[0, 1)$ be any real number, and let $\alpha_\eps\in[0, 1]$ be the unique real number such that $\alpha_\eps(1 + \alpha_\eps) = \frac{1 - \eps}{1 + \eps}$. Let $\cA = \{A_1, \ldots, A_r\}$ be a collection of PSD matrices. We define the \emph{connectivity threshold} of $\cA$ to be $N^*_\eps(\cA):= N(\alpha_\eps; \cA)$, where $N(\cdot)$ is the connectivity parameter. 
\end{definition}

\subsection{Main Theorem on Sparsifying Sums of PSD Matrices}

We can now state our main theorems on sparsifying sums of PSD matrices with sparsity as dictated by the connectivity threshold defined above. The following theorems give a general version of \cref{thm:mainthmintro} from the introduction. 

The first theorem gives us a lower bound against sparsification in terms of $N^*(\cA)$:
\begin{theorem}
\label{thm:mainthmlowerbound}
    Let $\cA = \{A_1, \ldots, A_r\}\subset\R^{n\times n}$ be a collection of PSD matrices, and let $\eps\in[0, 1)$ be any real number. Then there exists a subset $T\subseteq[r]$ of size $N_\eps^*(\cA) - 1$ that does not admit any non-trivial $\eps$-sparsifier. In particular, by letting $\eps\to 1^-$, we obtain that there exists a subset $T\subseteq[r]$ of size $\lim_{\alpha\to 0^+}N(\alpha;\cA) - 1$ which does not admit a $\eps$-sparsifier for any $\eps\in[0, 1)$.
\end{theorem}

We now state our main theorem that gives a complementary upper bound:
\begin{theorem}
\label{thm:mainthmupperbound}
    Let $\cA = \{A_1, \ldots, A_r\}\subset\R^{n\times n}$ be a collection of PSD matrices, and let $\eps\in[0, 1)$ be any real number. Let $T\subseteq[r]$ be any subset such that $\{A_i:i\in T\}$ are not all $0$, and let $m = \dim(\ker(\sum_{i\in T}A_i)^\perp)$.

    Then there exists a map $\mu:T\to\R_{\geq 0}$ such that 
        \[(1 - \eps)\sum_{i\in T}A_i\preceq\sum_{i\in T}\mu_iA_i\preceq(1 + \eps)\sum_{i\in T}A_i\mcom\]
        and 
        \[|\supp(\mu)|\leq O\left(\eps^{-2}(\log |T|)(\log m)\min_{\alpha\in(0, 1]}\frac{N(\alpha; \{A_i:i\in T\})}{\alpha}\right)\] 
        \[\leq O\left(\eps^{-2}(\log r)(\log n)\min_{\alpha\in(0, 1]}\frac{N(\alpha; \cA)}{\alpha}\right)\mper\]
        If $\eps\leq 0.99$, then one can substitute $\alpha = \alpha_\eps$ in the above expression to obtain that 
        \[|\supp(\mu)|\leq O\left(\eps^{-2}\cdot (\log r)(\log n)\cdot N^*_\eps(\cA)\right)\mper\]
        Furthermore, there is a randomized algorithm which in $\poly(r, n)$ time outputs weights $\mu$ satisfying the above properties.
\end{theorem}

\begin{remark}
    Note that in both \cref{thm:mainthmupperbound,thm:mainthmlowerbound} we pass to a subset `$T$' of our family $\cA$. This is necessary if we wish to develop a theory of PSD sparsification which admits matching (up to factors of $\poly(\eps^{-1})$, $\log r$, $\log n$) upper and lower bounds. Indeed, let $\cA = \{A_1, \ldots, A_k\}$ be a system of PSD matrices which can't be $\eps$-sparsified for any $\eps\in[0, 1]$.~\footnote{For example, one can take $A_1, \ldots, A_k$ to be the edge Laplacians of a tree on $k + 1$ vertices} In particular, if $M(\cA)$ is some purported ``reasonable'' measure which characterizes the sparsifiability of $\cA$, then $M(\cA)$ must grow with $|\cA|$. Now consider the family $\cA' = \cA\sqcup\{\sum_{i = 1}^k A_i\} = \{A_1, \ldots, A_k, \sum_{i = 1}^k A_i\}$. Note that $\cA'$ admits a $\eps$-sparsifier of size $1$ for all $\eps\geq 0$; indeed, assign weight $2$ to the matrix $\sum_{i = 1}^k A_i\in\cA'$, and zero out everything else in $\cA'$. Consequently, $M(\cA')$ should be a constant.

    Consequently, by adding a single element to $\cA$, the value of $M(\cdot)$ went from super-constant to constant, thus showing that any such $M(\cdot)$ can not be ``robust''.

    Note that we avoid this problem by allowing ourselves the flexibility to pass to subsets: This way, even for $\cA'$, we recognize that it has a ``unsparsifiable core'' $\cA$, and thus our measure $N^*(\cdot)$ increases by $\leq 1$ on the addition of a single matrix to the family. 

    Finally, note that $\cA'$ also shows that \cref{thm:mainthmlowerbound} is tight: Indeed, observe that $N(\alpha;\cA') = k + 1$ for all $\alpha\in(0, 1]$, and in particular $N^*(\cA') = k + 1$. Consequently, $\cA'$ is a set of size $k + 1 = N^*(\cA')$ which admits a $\eps$-sparsifier for all $\eps\in[0, 1]$, as desired.
\end{remark}

Consequently, by combining \cref{thm:mainthmlowerbound,thm:mainthmupperbound} for $\eps\in[0, 0.99]$, the connectivity threshold $N^*_\eps(\cA)$ characterizes the extent to which (any subset of) $\cA$ can be $\eps$-sparsified.

\begin{proof}[Proof of \cref{thm:mainthmintro}]
This follows immediately from the above two theorems by specializing to the case of $\eps=1/2$ and observing that $\alpha_{1/2} = (\sqrt{7/3}-1)/2$.
\end{proof}

We prove \cref{thm:mainthmlowerbound} in \cref{prf:mainthmlowerbound}. Proving the upper bound above is trickier and we first describe this. 

For \cref{thm:mainthmupperbound}, we use the technique of \emph{leverage score sparsification}, which states that \emph{any} system of PSD matrices $A_1, \ldots, A_r\in\R^{n\times n}$ can be, up to factors of $\poly(\eps^{-1})$ and $\log n$, sparsified down to $\sum_{i = 1}^r\|\tA_i\|_2$ entries, where $\tA_i = A^{-1/2}A_iA^{-1/2}$ is the ``normalization'' of the matrices by the sum $A = \sum_{i = 1}^r A_i$ to ensure that $\sum_{i = 1}^r\tA_i = \Id$. We reproduce this (well-known) statement in \cref{lem:psdsparsification}, as stated below:
\begin{lemma}
\torestate{
\label{lem:psdsparsification}
    Let $A_1, \ldots, A_r\in\R^{n\times n}$ be PSD matrices such that $A:= \sum_{i = 1}^r A_i$ is invertible. For any $\eps \in [0, 1]$ there exist weights $\mu:[r]\to\R_{\geq 0}$ such that for $A':= \sum_{i\in[r]}\mu_iA_i$ we have $A'\approx_\eps A$, and 
    \[|\!\supp(\mu)|\leq O\left(\frac{\log n}{\eps^2}\sum_{i}\|\tA_i\|_2\right)\mcom\]
    
    Furthermore, we can compute such a $\mu$ in $\poly(r, n)$ time with probability $\geq 1 - 1/\poly(n)$.
}
\end{lemma}

Thus all that remains is to show that $\sum_{i = 1}^r\|\tA_i\|_2$ can be upper bounded by the connectivity parameter as described in \cref{thm:mainthmupperbound}, which is the content of \cref{lem:sumofnormsupperbound}. To show that $\sum_{i = 1}^r\|\tA_i\|_2$ is small, we show via a sampling argument that if there are too many matrices $\{\tA_i\}_{i\in S}$ with too large a spectral norm, then one can extract a large $\alpha$-minimal subset $\{A_i:i\in S'\}$ for some $S'\subseteq S$, which then automatically furnishes an upper bound in terms of $N(\alpha;\cA)$. 
\begin{lemma}
\torestate{
    \label{lem:sumofnormsupperbound}
    Let $A_1, \ldots, A_r\in\R^{n\times n}$ be PSD matrices such that $A:= \sum_{i = 1}^r A_i$ is invertible. Write $\tA_i:= A^{-1/2}A_iA^{-1/2}$ for all $i\in[r]$. Then we have
    \[\sum_{i = 1}^r\|\tA_i\|_2\leq 4(1 + \log r)\cdot\min_{\alpha\in(0, 1]}\frac{N(\alpha;\cA)}{\alpha}\mper\]
}
\end{lemma}

We now finish the proof of \cref{thm:mainthmupperbound} modulo the proofs of the lemmas stated above.

\begin{proof}[Proof of \cref{thm:mainthmupperbound} assuming \cref{lem:psdsparsification,lem:sumofnormsupperbound}]
    Write $A = \sum_{i\in T}A_i$. We claim that it suffices to prove the sparsification statement for $\{A_i\vert_W\}_{i\in T}$, where $W:= \ker(A)^\perp$.
    
    Indeed, consider any vector $v\in\R^n$, and write $v = v_1 + v_2$, where $v_1\in W, v_2\in W^\perp = \ker(A)\subseteq\ker(A_i)$, where the last inclusion follows since $A_i$s are PSD matrices. Then $v^*A_iv = v_1^*A_iv_1 + v_2^*A_iv_2 = v_1^*A_iv_1 = v_1^*A_i\vert_Wv_1$ for all $i\in T$, and thus the only non-zero contributions to the quadratic form $v^*A_iv$ comes from the component of $v$ along $W$. Thus WLOG we can restrict our attention to $W$, as desired. 
    
    The second statement then follows by \cref{lem:psdsparsification,lem:sumofnormsupperbound} and noting that 
    \[\alpha_\eps = \frac{1}{2}\left(-1 + \sqrt{1 + 4\frac{1 - \eps}{1 + \eps}}\right) = \Omega(1)\] 
    for $\eps\leq 0.99$, and thus $\min_{\alpha\in(0, 1]}\frac{N(\alpha;\cA)}{\alpha}\leq O(N(\alpha_\eps;\cA)) = O(N^*_\eps(\cA))$. 
\end{proof}

We now focus on proving \cref{lem:psdsparsification,lem:sumofnormsupperbound}, to finish the proof of \cref{thm:mainthmupperbound}.

Let $A = \sum_{i = 1}^r A_i$. As noted in the proof of \cref{thm:mainthmupperbound}, WLOG we can restrict our attention to $\ker(A)^\perp$, on which $A$ is invertible. Thus from this point on we'll assume $A = \sum_{i = 1}^r A_i$ is invertible. 

Write $\tA_i:= A^{-1/2}A_iA^{-1/2}$,\footnote{For any positive definite $A$ with diagonalization $A=U\diag(\lambda_1, \ldots, \lambda_n)U^*$, define $A^{-1/2}:= U\diag(\lambda_1^{-1/2}, \ldots, \lambda_n^{-1/2})U^*$.} and note that $\sum_{i = 1}^r \tA_i = \Id$, and consequently $\|\tA_i\|_2\leq 1$ for all $i\in[r]$, since the matrices $\{\tA_i\}_{i\in[r]}$ are PSD too.

\begin{proof}[Proof of \cref{lem:psdsparsification}]
We sparsify by leverage score sampling. 

Thus write $R:= \eps^2/(16\log(4n))$, and for all $i\in[r]$ define
\begin{align}
\label{leveragescoredef}
    p_i:= \frac{1}{R}\|\tA_i\|_2\mcom
\end{align}
    where recall that $\tA_i:= A^{-1/2}A_iA^{-1/2}$, and $\sum_i\tA_i = \Id$.
Now, for each $i\in[r]$ define the random matrix $X_i$ (where the randomness in defining each of the $X_i$s is independent) as:
\begin{enumerate}[(1)]
    \item If $p_i\geq 1$, deterministically set $X_i:= \tA_i$.
    \item If $p_i < 1$, set $X_i:= \tA_i/p_i$ with probability $p_i$, and $0$ otherwise.
\end{enumerate}
Note that $\E[X_i] = \tA_i$ for all $i\in[r]$, and thus if $X = \sum_{i\in[r]}X_i$, then $\E[X] = \Id$. Consequently, we have 
\[(1 - \eps)\Id\preceq X\preceq(1 + \eps)\Id\iff -\eps\Id\preceq X - \E[X]\preceq\eps\Id\iff\|X - \E[X]\|\leq\eps\mper\]
Thus we define $Y_i:= X_i - \tA_i$ for all $i\in[r]$, $Y:= \sum_{i\in[r]}Y_i$, and note that $\|X - \E[X]\|\leq\eps\iff\|Y\|\leq\eps$, since $\E[Y] = 0$.

Now, note that for any $i\in[r]$, we have $\|Y_i\|\leq R$ with probability $1$. Also, $\E[Y_i^2] = \tA_i^2\left(\frac{1}{p_i} - 1\right)\preceq\tA_i^2/p_i\preceq R\tA_i$. Consequently, $\sum_{i = 1}^r\E[Y_i^2]\preceq R\sum_{i = 1}^r\tA_i = R\cdot\Id$, and thus $\|\sum_{i}\E[Y_i^2]\|\leq R$. Applying \cref{fact:matrixbernstein} yields that $\|Y\|\leq \eps$ with probability $\geq 1 - 1/\poly(n)\geq\frac{1}{2}$.

But note that $\|Y\|\leq\eps\iff(1 - \eps)\Id\preceq X\preceq(1 + \eps)\Id\iff(1 - \eps)A\preceq A^{1/2}XA^{1/2}\preceq(1 + \eps)A$, where $A = \sum_i A_i$. But note that $A^{1/2}XA^{1/2}$ can be written as $\sum_i\mu_iA_i$ for a non-negative map $\mu:[r]\to\R_{\geq 0}$. 

Note that $|\supp(\mu)|$ can be written as a sum of independent $\{0, 1\}$-valued random variables such that $\E|\supp(\mu)| = \sum_ip_i = \frac{1}{R}\sum_{i}\|\tA_i\|_2\geq \frac{1}{R}\|\sum_{i}\tA_i\|_2 = \frac{1}{R}\geq 16\log(4n)$. Consequently, by a usual Chernoff bound, $\Pr(|\supp(\mu)| < 2\E|\supp(\mu)|)\geq 1 - e^{-\E|\supp(\mu)|/3}\geq 1 - 1/(4n)^5$. 
    
Thus with probability $\geq \frac{1}{2} - \frac{1}{(4n)^5} > 0$, we have that $\sum_i\mu_iA_i\approx_\eps A$, and 
    \[|\supp(\mu)|\leq 2\cdot \frac{1}{R}\sum_{i}\|\tA_i\|_2 \leq O\left(\frac{\log n}{\eps^2}\sum_{i}\|\tA_i\|_2\right)\mcom\]
as desired. Furthermore, $\mu$ can clearly be computed in $\poly(r, n)$ time.
\end{proof}

We now focus on proving \cref{lem:sumofnormsupperbound}.

For every $i\in[r]$, let $w_i$ be a unit vector at which $\tA_i$ attains its spectral norm, i.e.\ $w_i^* \tA_iw_i = \|\tA_i\|_2$. Note that $w_i^* \left(\sum_{i\in [r]}\tA_i\right)w_i = w_i^* w_i = 1$. Finally, for any $\eta\geq 0$, define $S_\eta:= \{i\in[r]: \|\tA_i\|_2\geq\eta\}$. 

The proof outline of \cref{lem:sumofnormsupperbound} is as follows: We claim that it suffices to show that for any $\alpha, \eta\in(0, 1]$, $|S_\eta|\leq O(N(\alpha)/\alpha\eta)$. Indeed, once we have this upper bound, we can simply integrate against $\eta$ to get the desired statement (the integral of $d\eta/\eta$ is what produces the $\log r$ term).

To actually obtain the upper bound for $|S_\eta|$, consider a $p$-subset of $S_\eta$, i.e.\ include every element of $S_\eta$ independently in this subset with probability $p:= \alpha\eta/2$. We then show that a constant fraction of this set violates the connectivity parameter property, i.e.\ for any given element of this subset, $\alpha$ times the matrix corresponding to it is \emph{not} dominated (in the PSD order $\preceq$) by the sum of other matrices in this subset with $\Omega(1)$ probability. Note that the collection of these elements is a $\alpha$-minimal set, and thus $N(\alpha) > \Omega(p|S_\eta|) = \Omega(\alpha\eta|S_\eta|)$, as desired.

\begin{proof}[Proof of \cref{lem:sumofnormsupperbound}]
Fix any $\alpha, \eta\in(0, 1]$. We claim that  
\[|S_\eta| < \frac{4N(\alpha;\cA)}{\alpha\eta}\mper\]
Given this claim, the assertion in the lemma follows by a simple bucketing argument: Indeed, write $D:= 4N(\alpha;\cA)/\alpha, \beta:= D/r$. Then
    \[\sum_{i\in [r]}\|\tA_i\|_2 = \sum_{i\in[r]}\int_0^{\|\tA_i\|}\!\mathrm{d}y = \int_0^1\!\left|\{i\in [r]: \|\tA_i\|\geq y\}\right|\,\mathrm{d}y = \int_0^1 |S_y|\,\mathrm{d}y\leq\int_0^1 \min\left\{r, \frac{D}{y}\right\}\,\mathrm{d}y\]
    \[ = \int_0^\beta r\,\mathrm{d}y + \int_\beta^1 \frac{D}{y}\,\mathrm{d}y = D\left(1 + \log\frac{r}{D}\right)\leq \frac{4N(\alpha;\cA)}{\alpha}\left(1 + \log r\right),\]
    as desired.

\parhead{Upper bounding $|S_\eta|$} Now, let $\{X_s\}_{s\in S_\eta}$ be i.i.d. $\operatorname{Bern}(p)$ random variables, where $\operatorname{Bern}(p)$ stands for a Bernoulli random variable with expectation $p$. Recall that we set $p = \alpha\eta/2$. For every $s\in S_\eta$, define the random variable
\[Y_s:= \1\left(\sum_{t\in S_\eta\setminus\{s\}}X_tw_s^*\tA_tw_s < \alpha\cdot X_sw_s^*\tA_sw_s\right).\]
Finally define the random set $S_\eta':= \{s\in S_\eta: Y_s = 1\}$. Note that $\E|S_\eta'| = \sum_{s\in S_\eta}\Pr(Y_s = 1)$. We will show that $|S'_\eta|$ is $\alpha$-minimal, and has a large expected size, which suffices to finish the proof.

Note that since the matrices $\{\tA_i\}$ are PSD, if $X_s = 0$ then $Y_s$ has to be $0$, or equivalently if $Y_s = 1$ then $X_s$ had to have been $1$. Thus $\Pr(Y_s = 1) = \Pr(Y_s = 1\mid X_s = 1)\cdot\Pr(X_s = 1) = p\cdot\Pr(Y_s = 1\mid X_s = 1)$. Now, since the random variables $\{X_s\}$ are independent, we have 
\[\Pr(Y_s = 1\mid X_s = 1) = \Pr\left(\sum_{t\in S_\eta\setminus\{s\}}X_tw_s^*\tA_tw_s < \alpha\cdot w_s^*\tA_sw_s\right) = \Pr\left(\sum_{t\in S_\eta\setminus\{s\}}X_tw_s^*\tA_tw_s < \alpha\cdot \|\tA_s\|_2\right)\]
\[\overset{\text{Markov's Inequality}}{\geq}1 - \frac{\E\sum_{t\in S_\eta\setminus\{s\}}X_tw_s^*\tA_tw_s}{\alpha\cdot \|\tA_s\|_2}\geq 1 - \frac{\sum_{t\in S_\eta\setminus\{s\}}\E[X_t]w_s^*\tA_tw_s}{\alpha\eta}\mcom\]
where the last inequality follows since $s\in S_\eta\implies\|\tA_s\|_2\geq\eta$.

Now, 
\[\sum_{t\in S_\eta\setminus\{s\}}\E[X_t]w_s^*\tA_tw_s = p\sum_{t\in S_\eta\setminus\{s\}}w_s^*\tA_t w_s\leq p\sum_{i = 1}^rw_s^*\tA_i w_s = p\mper\]
where the last equality follows since $\sum_{i = 1}^r\tA_i = \Id$, and $\|w_s\| = 1$. Consequently, 
\[\Pr(Y_s = 1\mid X_s = 1)\geq 1 - \frac{p}{\alpha\eta} = \frac{1}{2}.\]
Putting everything together, we obtain that $\Pr(Y_s = 1)\geq p/2\implies\E|S_\eta'|\geq p|S_\eta|/2$. Since $\E|S_\eta'|\geq p|S_\eta|/2$, there exists a choice of $\{X_s\}_{s\in S_\eta}$ such that $|S_\eta'|\geq p|S_\eta|/2$. Now, consider any $s\in S_\eta'$. Since $Y_s = 1$, we have $X_s = 1$, and thus
\[\alpha w_s^*\tA_sw_s > \sum_{t\in S_\eta\setminus\{s\}}X_tw_s^*\tA_tw_s = \sum_{t\in S_\eta\setminus\{s\}, X_t = 1}w_s^*\tA_tw_s\overset{(\ast)}{\geq} \sum_{t\in S_\eta\setminus\{s\}, Y_t = 1}w_s^*\tA_tw_s = \sum_{t\in S'_\eta\setminus\{s\}}w_s^*\tA_tw_s\mper\]
Here $(\ast)$ follows since $\{t\in S_\eta: Y_t = 1\}\subseteq\{t\in S_\eta: X_t = 1\}$. 

Consequently, if we write $v_s = A^{-1/2}w_s$, then $w_s^*\tA_xw_s = v_s^* A_xv_s$ for all $x\in S$, and thus 
\[\alpha v_s^*A_sv_s > \sum_{t\in S'_\eta\setminus\{s\}}v_s^*A_tv_s\implies\alpha A_s\not\preceq\sum_{t\in S'_\eta\setminus\{s\}}A_t\mper\]
Since $\alpha A_s\not\preceq\sum_{t\in S'_\eta\setminus\{s\}}A_t$ for all $s\in S'_\eta$, we must have $|S'_\eta| < N(\alpha;\cA)$ by the definition of $N(\alpha;\cA)$. Consequently, we have $N(\alpha;\cA) > |S_\eta'|\geq p|S_\eta|/2 = \alpha\eta|S_\eta|/4$ thus implying 
\[|S_\eta| < \frac{4N(\alpha;\cA)}{\alpha\eta}\mcom\]
as desired.
\end{proof}

\subsection{Proof of Lower Bound on Sparsity}
\label{prf:mainthmlowerbound}

We now prove \cref{thm:mainthmlowerbound}. The proof relies on the definition of connectivity threshold and definition of sparsification. 

\begin{proof}[Proof of \cref{thm:mainthmlowerbound}]
    Let $T$ be a $\alpha_\eps$-minimal set, i.e.\ for all $i\in T$, $\alpha_\eps A_i\not\preceq\sum_{j\in T\setminus\{i\}}A_j$. Consequently, for every $i\in T$, there is a vector $w_i$ such that $\alpha_\eps w_i^*A_iw_i > \sum_{j\in T\setminus\{i\}}w_i^*A_jw_i\geq 0$. Assume for the sake of contradiction that $\mu:T\to\R_{\geq 0}$ is a $\eps$-sparsifier of $T$, i.e.\ $|\supp(\mu)| < |T|$, and $(1 - \eps)\sum_{j\in T}A_j\preceq\sum_{j\in T}\mu_jA_j\preceq(1 + \eps)\sum_{j\in T}A_j$.

Since $|\supp(\mu)| < |T|$, there exists some $k\in T$ such that $\mu_k = 0$. Since $\mu$ is a $\eps$-sparsifier of $T$, we have 
\[\sum_{j\in T}\mu_jA_j\succeq(1 - \eps)\sum_{j\in T}A_j\] 
\[\implies\sum_{j\in T\setminus\{k\}}\mu_jw_k^*A_jw_k = \sum_{j\in T}\mu_jw_k^*A_jw_k\geq(1 - \eps)\sum_{j\in T}w_k^*A_jw_k\geq(1 - \eps)w_k^*A_kw_k\mper\] 
On the other hand, 
\[\sum_{j\in T\setminus\{k\}}\mu_jw_k^*A_jw_k\leq\|\mu\|_\infty\cdot\sum_{j\in T\setminus\{k\}}w_k^*A_jw_k < \|\mu\|_\infty\cdot\alpha_\eps w_k^*A_kw_k\mcom\]
where $\|\mu\|_\infty:= \max_{j\in T}\mu_j$. Thus 
\[\|\mu\|_\infty\cdot\alpha_\eps w_k^*A_kw_k > (1 - \eps)w_k^*A_kw_k\mper\]
Since $w_k^*A_kw_k \geq 0$, we must have $\alpha_\eps\|\mu\|_\infty > 1 - \eps$. Now suppose $\ell\in T$ is such that $\mu_\ell = \|\mu\|_\infty$. Then 
\[\mu_\ell w_\ell^*A_\ell w_\ell\leq\sum_{j\in T}\mu_j w_\ell^*A_j w_\ell\overset{(\ast)}{\leq} (1 + \eps)\sum_{j\in T} w_\ell^*A_j w_\ell = (1 + \eps)\left(w_\ell^*A_\ell w_\ell + \sum_{j\in T\setminus\{\ell\}} w_\ell^*A_j w_\ell\right)\] 
\[< (1 + \eps)(1 + \alpha_\eps)w_\ell^*A_\ell w_\ell\implies(1 + \eps)(1 + \alpha_\eps) > \mu_\ell = \|\mu\|_\infty\mper\]
Here $(\ast)$ holds since $\mu$ is a $\eps$-sparsifier.

Consequently, $\alpha_\eps(1 + \alpha_\eps)(1 + \eps) > 1 - \eps$, which leads to a contradiction since $\alpha_\eps(1 + \alpha_\eps) := (1 - \eps)/(1 + \eps)$. 
\end{proof}

\section{Cayley Sparsifiers of Cayley Graphs}
In this section we obtain $\eps$-sparsifiers of size $O(\eps^{-2}\log^4 N)$ for Cayley graphs, as stated below.
\begin{theorem}[Cayley Sparsifiers of Cayley Graphs]
\label{thm:cayleysparsification}
    Let $\cG$ be a group of size $N$, let $S\subseteq \cG$ be a symmetric subset which generates $\cG$. Then there exists a symmetric weight function $\mu:S\to\R_{\geq 0}$ such that $|\supp(\mu)|\leq O(\eps^{-2}\log^3 N\log |S|)\leq O(\eps^{-2}\log^4N)$, and $L'\approx_\eps L$, where $L$ is the Laplacian of $\Cay(\cG, S)$, and $L'$ is the Laplacian of $\Cay(\cG, S, \mu)$. Furthermore, we can compute such a $\mu$ in $\poly(N)$ time with probability $\geq 1 - \exp(-N)$.
\end{theorem}
In light of \cref{thm:mainthmupperbound}, it suffices to upper bound $\min_{\alpha\in(0, 1]}N(\alpha;\cL)/\alpha$ for the collection of Laplacians $\cL:= \{L_s:s\in S\}$ associated to $\Cay(\cG, S)$. We do that in \cref{lem:cayleycyclethreshold}:
\begin{lemma}
\torestate{
    \label{lem:cayleycyclethreshold}
    Let $\cG$ be any finite group of size $N$, and let $S\subseteq \cG$ be any symmetric subset. Then
    \[N\left(\alpha;\{L_s\}_{s\in S}\right)\leq\lceil\log_2 N\rceil + 1\mcom\]
    for $\alpha:= \frac{1}{8(\lceil\log_2 N\rceil + 1)}$. If we have a symmetric weight function $\nu:S\to\R_{\geq 0}$, then $N\left(\alpha/W;\{\nu(s)L_s\}_{s\in S}\right)\leq \lceil\log_2 N\rceil + 1$, where
    \[W:= \frac{\max_{s\in S}\nu(s)}{\min_{s\in S:\nu(s) > 0}\nu(s)}\mper\]
}
\end{lemma}
Note that once we prove \cref{lem:cayleycyclethreshold}, the proof of \cref{thm:cayleysparsification} is immediate by combining \cref{lem:cayleycyclethreshold} with \cref{thm:mainthmupperbound}. We can always amplify the success probability of our algorithm as we can check whether or not we have a sparsifier in $\poly(N)$ time. 

We next briefly explain how \cref{lem:cayleycyclethreshold} is proved: The proof of \cref{lem:cayleycyclethreshold} is very similar to the proof of \cref{lem:cayleyconnparam}, in that it first creates a relation $x = x_1\cdots x_k$ amongst the generators in $S$, and then converts the relation into a statement of the form $\alpha L_x\preceq\sum_{y\in S\setminus\{x\}}L_y$, where $\alpha = \Omega(\frac{1}{km})$. Here $k$ is the length of the relation, and $m$ is the maximum multiplicity of any element in the multiset $\{x_1, \ldots, x_k\}$. In \cref{lem:cayleyconnparam}, where we have just $m(\cG) + 1$ elements, all we can guarantee is that the relation length $k$ will be $O(N)$, and the multiplicity of any element in the relation will also be $O(N)$ (trivially), thus yielding $\alpha\sim 1/N^2$. In \cref{lem:cayleycyclethreshold}, where we have $\lceil\log_2 N\rceil + 1$ elements, a simple pigeonhole principle argument shows that we get a relation of length $O(\log N)$ where any element appears $\leq O(1)$ times, which yields $\alpha\sim 1/\log N$, as desired.

The proof of \cref{lem:cayleyconnparam} uses a simple combinatorial ``cycle-decomposition'' argument to prove the above statement. This same technique also works to prove \cref{lem:cayleycyclethreshold}. However, in this section, we give a different, more linear-algebraic proof of the same fact, using \cref{unitarytriang} as an auxiliary result.

Thus we finish the proof of \cref{thm:cayleysparsification} using \cref{lem:cayleycyclethreshold}:
\begin{proof}[Proof of \cref{thm:cayleysparsification} using \cref{lem:cayleycyclethreshold,thm:mainthmupperbound}]
    Note that $\Cay(\cG, S)$ is connected since $S$ generates $\cG$. Consequently, $\dim(\ker(L)^\perp) = N - 1$, and we can directly invoke \cref{thm:mainthmupperbound} and \cref{lem:cayleycyclethreshold},\footnote{\label{symmetrizationfootnote}Suppose the sparsifier we get is $\mu:S\to\R_{\geq 0}$, and the corresponding Laplacian is $\sum_{s\in S}\mu_sL_s = \sum_{s\in S}\left(\frac{\mu_s + \mu_{s^{-1}}}{2}\right)L_s$, where the equality follows since $L_s = L_{s^{-1}}$. Thus we can relabel $\mu_{\text{new}}(s):= \frac{\mu_{\text{old}}(s) + \mu_{\text{old}}(s^{-1})}{2}$, and notice that $\mu_{\text{new}}$ is symmetric. Also note that $|\supp(\mu_{\text{new}})|\leq 2|\supp(\mu_{\text{old}})|$, and we absorb the factor of $2$ in $O(\cdot)$ of the upper bound for $|\supp(\mu)|$} and note that 
    \[\min_{\alpha\in(0, 1]}\frac{N(\alpha)}{\alpha}\leq \frac{N\left(\frac{1}{8(\lceil\log_2 N\rceil + 1)}\right)}{\frac{1}{8(\lceil\log_2 N\rceil + 1)}}\leq\frac{(\lceil\log_2 N\rceil + 1)}{\frac{1}{8(\lceil\log_2 N\rceil + 1)}}\leq O(\log^2 N)\mper\qedhere\]
\end{proof}
For weighted Cayley graphs, a simple dyadic division of the weights combined with the same proof technique as above yields a sparsifier.
\begin{theorem}[Cayley Sparsifiers of weighted Cayley Graphs]
\label{thm:cayleysparsificationweighted}
    Let $\cG$ be a group of size $N$, let $S\subseteq \cG$ be a symmetric subset which generates $\cG$. If we have a non-zero symmetric weight function $\nu:S\to\R_{\geq 0}$, then there exist symmetric weights $\mu:S\to\R_{\geq 0}$ such that 
    \[|\supp(\mu)|\leq O(\eps^{-2}\log^3 N\log |S|\log W)\leq O(\eps^{-2}\log^4 N\log W)\mcom\] 
    and $L'\approx_\eps L$, where $L$ is the Laplacian of $\Cay(\cG, S, \nu)$, $L'$ is the Laplacian of $\Cay(\cG, S, \mu)$, and 
    \[W:= \frac{\max_{s\in\supp(\nu)}\nu(s)}{\min_{s\in\supp(\nu)}\nu(s)}\mper\] Furthermore, we can compute such a $\mu$ in $\poly(N)$ time with probability $\geq 1 - 1/\poly(N)$.
\end{theorem}
\begin{proof}[Proof of \cref{thm:cayleysparsificationweighted} using \cref{lem:cayleycyclethreshold,thm:mainthmupperbound}]
    Multiplicatively bucket the generators, i.e.\ write $S = S_1\sqcup S_2\sqcup\cdots\sqcup S_r$, where for all $i\in[r]$, and all $s, s'\in S_i$, we have $\nu(s)\leq 2\nu(s')$. Note that we can take $r\leq\lceil\log_2 W\rceil + 1$, and we can also take $S_i$s to be symmetric. Also note that $L = \sum_{s\in S}\nu(s)L_s = \sum_{i = 1}^r\sum_{s\in S_i}\nu(s)L_s =: \sum_{i = 1}^rL_i$, where $L_i:= \sum_{s\in S_i}\nu(s)L_s$. Thus if we have $\eps$-sparsifiers $L'_1, \ldots, L'_r$ of $L_1, \ldots, L_r$ respectively, then $(1 - \eps)(L_1 + \cdots + L_r)\preceq L'_1 + \cdots + L'_r\preceq (1 + \eps)(L_1 + \cdots + L_r)\implies L'\approx_\eps L$, where $L':= L'_1 + \cdots + L'_r$, and thus $L'$ becomes a sparsifier of $L$. But note that since the weights in $\nu|_{S_i}$ differ by a factor of $\leq 2$, by \cref{thm:mainthmupperbound} and \cref{lem:cayleycyclethreshold}, we're done.
\end{proof}
Thus we now focus on proving \cref{lem:cayleycyclethreshold}:
\subsection{Proof of \cref{lem:cayleycyclethreshold}}
Before we establish \cref{lem:cayleycyclethreshold}, we prove a small linear algebraic lemma, stating that if $U_1, U_2, \ldots, U_r$ are \emph{almost invariant} unitary matrices, i.e.\ $\|U_iw - w\|$ is small for all $w$ and all $i\in[r]$, then $U_1\cdots U_r$ is also almost invariant:
\begin{lemma}
\label{unitarytriang}
    Let $w$ be a vector, and let $U_1, \ldots, U_r$ be unitary matrices such that $\|w\|^2 - w^*H_{U_i}w\leq\delta_i$ for some $\delta_i\geq 0$ for all $1\leq i\leq r$. Then 
    \[\|w\|^2 - w^*H_{U_1\cdots U_r}w\leq (\sqrt{\delta_1} + \cdots + \sqrt{\delta_r})^2.\]
\end{lemma}
\begin{proof}
    Since $w^*H_{U_i}w\geq\|w\|^2 - \delta_i$, we have 
    \[\|U_iw - w\|^2 = \langle U_iw, U_iw\rangle + \langle w, w\rangle - \langle w, U_iw\rangle - \langle U_iw, w\rangle = 2\|w\|^2 - 2w^*H_{U_i}w\leq 2\delta_i.\]
    Consequently, 
    \[\left\|U_1\cdots U_rw - w\right\| = \left\|\sum_{i = 1}^{r}\left( \prod_{j = 1}^{r - i + 1} U_jw - \prod_{j = 1}^{r - i} U_jw\right)\right\|\leq\sum_{i = 1}^{r}\left\|\left( \prod_{j = 1}^{r - i + 1} U_jw - \prod_{j = 1}^{r - i} U_jw\right)\right\|\]
    \[= \sum_{i = 1}^{r}\left\| U_{r - i + 1}w - w\right\|\leq\sum_{i = 1}^r\sqrt{2\delta_i}.\]
    In the above inequality $\prod_{j = 1}^{0}U_j$ should be read as the empty product, evaluating to $\Id$. In the last equality we also use the fact that unitaries are isometries, i.e.\ for any unitary $U$ and any $x$, $\|Ux\|_2 = \|x\|_2$. Finally, 
    \[\|w\|^2 - w^*H_{U_1\cdots U_r}w = \frac{1}{2}\|U_1\cdots U_rw - w\|^2\leq (\sqrt{\delta_1} + \cdots + \sqrt{\delta_r})^2, \]
    as desired.
\end{proof}
We can now finish the proof of \cref{lem:cayleycyclethreshold}:
\begin{proof}[Proof of \cref{lem:cayleycyclethreshold} using \cref{unitarytriang}]
       Consider any subset $T\subseteq \cG$ of size $\lceil\log_2 N\rceil + 1$. Since $|T| = \lceil\log_2 N\rceil + 1$, $T$ has $> N$ subsets. In particular, two distinct subsets $T_1, T_2\subseteq T$ must have the same product (since $\cG$ has only $N$ elements). Since $T_1\neq T_2$, there must exist some $s\in (T_1\setminus T_2)\cup(T_2\setminus T_1)$. Consequently, rearranging terms in the product, we must have $s = t_1\cdots t_r$, where $t_1, \ldots, t_r\in (T\setminus\{s\})\cup (T\setminus\{s\})^{-1}$, and $r \leq 2|T|$. 
    
    Now, since $\scrR(\cdot)$ is a representation, we have $\scrR(s) = \scrR(t_1)\cdots \scrR(t_r)$. Also let $w$ be an arbitrary vector. To show $\alpha L_s\preceq\sum_{t\in T\setminus\{s\}}L_t$, we have to show that $\alpha w^*L_sw\leq\sum_{t\in T\setminus\{s\}}w^*L_tw$. Note that we may assume $w^*L_sw > 0$ as otherwise the inequality is trivial.
    
    Then note that $w^*L_xw = w^*(\Id - H_{\scrR(x)})w = \|w\|^2 - w^*H_{\scrR(x)}w$ for all $x\in \cG$. Let $\delta_1, \ldots, \delta_r\geq 0$ be real numbers such that $w^*L_{t_i}w = \|w\|^2 - w^*H_{\scrR(t_i)}w = \delta_i w^*L_sw$ for all $i\in[r]$.
    
    Consequently, by \cref{unitarytriang}, 
    \[w^*L_sw = \|w\|^2 - w^*H_{\scrR(s)}w\leq (\sqrt{\delta_1} + \cdots + \sqrt{\delta_r})^2 w^*L_sw\implies\sqrt{\delta_1} + \cdots + \sqrt{\delta_r}\geq 1\]
    \[\overset{\text{Cauchy-Schwarz}}{\implies}\delta_1 + \cdots + \delta_r\geq\frac{1}{r}\implies\sum_{i = 1}^rw^*L_{t_i}w\geq\frac{1}{r}w^*L_sw\mper\]
    Now, note that $L_t = L_{t^{-1}}$ for all $t\in \cG$. Also if $x, y\in \cG$ are such that $y\not\in\{x, x^{-1}\}$, then $L_x\neq L_y$.\footnote{One way to see this is that $\Cay(\cG, \{x, x^{-1}\}) = \Cay(\cG, \{y, y^{-1}\})$ iff $y\in\{x, x^{-1}\}$. Here when we write ``$=$'', we mean equality as labeled graphs: Indeed, note that in $\Cay(\cG, \{x, x^{-1}\})$, $\id_{\cG}$ is connected to $x, x^{-1}$ while in $\Cay(\cG, \{y, y^{-1}\})$ it is connected to $y, y^{-1}$} Consequently, the multiplicity of $L_{t_i}$ in the multiset $\{L_{t_j}:j\in[r]\}$ is $\leq 4$: Indeed, $t_i, t_i^{-1}$ can both occur at most twice in $T_1, T_2$. Consequently, $\sum_{t\in T\setminus\{s\}}w^*L_{t}w\geq\frac{1}{4}\sum_{i = 1}^rw^*L_{t_i}w$, and thus 
    \[\sum_{t\in T\setminus\{s\}}w^*L_{t}w\geq\frac{1}{4r}\cdot w^*L_sw\geq\frac{1}{8|T|}\cdot w^*L_sw\mcom\]
    as desired.

    Finally, we can deal with the weight function $\nu$ as follows: Consider any $T$ of size $\lceil\log_2 N\rceil + 1$. If we have $\nu(t) = 0$ for any $t\in T$, then we're already done, since $0 = (\alpha/W)\cdot \nu(t)L_t\preceq\sum_{t'\in T\setminus\{t\}}L_{t'}$. Thus assume $\nu(t) > 0$ for all $t\in T$. By running the same proof as above, we can obtain some $s\in T$ such that $\alpha L_s\preceq\sum_{t\in T\setminus\{s\}}L_t$. But then $\alpha \nu(s)L_s\preceq W\cdot\sum_{t\in T\setminus\{s\}}\nu(t)L_t$, and we're done.
\end{proof}

\subsection{Connectivity Thresholds for Cayley Graphs}
In this section, we'll show that for any abelian group $\cG$, $N^*(\cL(\cG)) = \Omega(\log|\cG|)$. We do so by first producing $\alpha$-minimal sets in cyclic groups $\Z_N$, and then lifting them to all abelian groups via \cref{fact:fundabel}. Then, using the result for abelian groups, we'll show that $N^*(\cL(\cG)) = \Omega(\sqrt{\log|\cG|})$ for \emph{any} finite group $\cG$ by embedding our $\alpha$-minimal sets for abelian groups into some abelian subgroup of $\cG$.

\begin{lemma}[$\alpha$-minimal sets in $\Z_N$]
\label{lem:alphaminimalzn}
    For any $N\geq 2$ and $\alpha\in(0, 1]$, we have $N(\alpha;\cL(\Z_N)) = \Omega((\log N)/\log(4/\alpha))$. 
\end{lemma}
\begin{proof}
    Let $k$ be the smallest even integer such that $\alpha\geq \frac{\pi^2}{2(k^2-1)}$, and set $r:= \lfloor \gamma\log_k N\rfloor$, where $\gamma:= 1/8$. Note that $r = \Omega((\log N)/\log(4/\alpha))$.
    
    Define $S\subseteq\Z_N$ as $S:= \{k, k^2, \ldots, k^r\}$. To show that $S$ is $\alpha$-minimal, we have to show that for any $i\in[r]$, $\alpha L_{k^i}\not\preceq\sum_{j\neq i}L_{k^j}$. Write $\chi:= \chi_{\lfloor N/2k^i\rfloor}, \delta:= N/2k^i - \lfloor N/2k^i\rfloor, \ell = 2\|\chi\|_2^2$, where $\chi$ is the vector defined in \cref{fact:abelianeig}. We'll show that $\alpha\chi^\top L_{k^i}\chi > \sum_{j\neq i}\chi^\top L_{k^j}\chi$, which proves the desired claim.

    Now, 
    \[\chi^\top L_{k^i}\chi = \sin^2\left(\frac{\pi k^{i}\cdot\lfloor\frac{N}{2k^i}\rfloor}{N}\right)\cdot(2\|\chi\|_2^2) = \sin^2\left(\frac{\pi}{2} - \frac{\pi k^{i}\delta}{N}\right)\cdot\ell = \cos^2\left(\frac{k^{i}\pi\delta}{N}\right)\cdot\ell\]
    \[\geq \cos^2\left(\frac{k^{r}\pi}{N}\right)\cdot\ell\overset{(\ast)}{\geq} \cos^2\left(\frac{\pi}{N^{1 - \gamma}}\right)\cdot\ell\geq\left(1 - \frac{\pi^2}{N^{2(1 - \gamma)}}\right)\cdot\ell\mcom\]
    where $(\ast)$ holds since $0\leq k^r\leq N^\gamma$, and the last inequality holds since $\cos^2(x)\geq 1 - x^2$ for all $x\in\R$. On the other hand, 
    \[\frac{1}{\ell}\sum_{j\neq i}\chi^\top L_{k^j}\chi = \sum_{j \neq  i}\sin^2\left(\frac{\pi k^{j}\cdot\lfloor\frac{N}{2k^i}\rfloor}{N}\right) = \sum_{j < i}\sin^2\left(\frac{\pi k^{j - i}}{2} - \frac{\pi k^{j}\delta}{N}\right) + \sum_{j > i}\sin^2\left(\frac{\pi k^{j - i}}{2} - \frac{\pi k^{j}\delta}{N}\right)\]
    \[ = \sum_{j < i}\sin^2\left(\frac{\pi k^{j - i}}{2} - \frac{\pi k^{j}\delta}{N}\right) + \sum_{j > i}\sin^2\left(\frac{\pi k^{j}\delta}{N}\right) \leq \sum_{j < i}\sin^2\left(\frac{\pi k^{j - i}}{2}\right) + \sum_{j > i}\sin^2\left(\frac{\pi k^{j}}{N}\right)\]
    \[\leq \sum_{j < i}\frac{\pi^2 k^{2(j - i)}}{4} + r\cdot\frac{\pi^2 k^{2r}}{N^2}\leq\pi^2\left(\frac{1}{4}\sum_{t = 1}^\infty k^{-2t} + \frac{rk^{2r}}{N^2}\right)\leq\pi^2\left(\frac{1}{4(k^2 - 1)} + \frac{\log_k N}{N^{2(1 - \gamma)}}\right)\mcom\]
    where throughout the above derivation we use the fact that $\sin^2(x)\leq x^2$ for all $x\in\R$, and $\sin^2(\pi k^{j - i}/2 - x) = \sin^2(x)$ for $j > i$ since $k$ is even. By assumptions on $k, \gamma, r$, we have $\alpha\left(1 - \frac{\pi^2}{N^{2(1 - \gamma)}}\right)\ell > \ell\pi^2\left(\frac{1}{4(k^2 - 1)} + \frac{\log_k N}{N^{2(1 - \gamma)}}\right)$ for all $N\geq 2$ such that $r\geq 2$ (note that if $r = 1$ then the statement is trivial), thus yielding that $\alpha\chi^\top L_{k^i}\chi > \sum_{j\neq i}\chi^\top L_{k^j}\chi$, as desired. \qedhere
\end{proof}
We can now lift \cref{lem:alphaminimalzn} to all finite abelian groups $\cG$ using \cref{fact:fundabel}.
\begin{lemma}\label{lem:abelianlb}
    For any abelian group $\cG$ of size $N$, and any $\alpha\in(0, 1]$, $N(\alpha;\cL(\cG))=\Omega((\log N)/\log(4/\alpha))$. In particular, $N^*(\cL(\cG))=\Omega(\log N)$ and $\min_{\alpha\in(0, 1]}\frac{N(\alpha;\cL(\cG))}{\alpha} = \Omega(\log N)$.
\end{lemma}
\begin{proof}
    By \cref{fact:fundabel}, we have $\cG\cong\bigoplus_{i = 1}^t\Z_{N_t}$, and thus $N = \prod_{i = 1}^tN_i$. Let $S_t\subseteq\Z_{N_t}$ be an $\alpha$-minimal subset of size $\Omega((\log N_t)/\log(4/\alpha))$, which exists by \cref{lem:alphaminimalzn}. Consider the subset $S := \bigcup_{i = 1}^t\{0\}\times\cdots\times S_i\times\cdots\times\{0\}$ in $\cG$, where in the $i^{\mathrm{th}}$ term in the union, $S_i$ appears in the $i^{\mathrm{th}}$ position of the product. Note that $|S| = \sum_{i = 1}^t|S_i|\geq\sum_{i = 1}^t\Omega((\log N_i)/\log(4/\alpha)) = \Omega((\log N)/\log(4/\alpha))$. Also note that for any $s\in S$, exactly one ``coordinate'' of $s$ is non-zero.

    We claim that $S$ is $\alpha$-minimal. Indeed, consider any $s\in S$, and suppose the $i^{\mathrm{th}}$ coordinate of $s$ is non-zero. Let $\chi_{s,i}\in\R^{\Z_{N_i}}$ be the vector from the proof of \cref{lem:alphaminimalzn} which witnesses that $\alpha L_{s_i}\not\preceq\sum_{t_i\in S_i\setminus\{s_i\}}L_{t_i}$. Note that by \cref{fact:abelianeig}, $\chi:= \chi_{s,i}\otimes\1\in\R^{\cG}$ is such that for any $t\in S$ with $t_i\neq 0$, $L_t\chi = \lambda_{t, i}\chi$, where $L_{t_i}\chi_{s, i} = \lambda_{t, i}\chi_{s, i}$, and for any other $t'\in S$ with $t'_i = 0$, $L_{t'}\chi = 0$. Consequently, $\chi$ witnesses that $\alpha L_{s}\not\preceq\sum_{t\in S\setminus\{s\}}L_{t}$, as desired.
\end{proof}
Note that the above result is tight (up to constants) for groups such as $\F_2^n$, for which we have $N(\alpha;\cL(\F_2^n)) = n + 1 = \log_2 N + 1$ for all $\alpha\in(0, 1]$.

We can extend the above result to \emph{all} groups (possibly non-abelian) by passing to an abelian subgroup. As an immediate corollary, we conclude that for every group, there is a Cayley graph that does not admit a spectral sparsifier with $o(\sqrt{\log N})$ generators. 

\begin{lemma}\label{lem:nonabelianlb}
    For any group $\cG$ (possibly non-abelian) of size $N$, and any $\alpha\in(0, 1]$, we have $N(\alpha;\cL(\cG)) = \Omega((\sqrt{\log N})/\log(4/\alpha))$. In particular, $N^*(\cL(\cG))=\Omega(\sqrt{\log N})$ and $\min_{\alpha\in(0, 1]}\frac{N(\alpha;\cL(\cG))}{\alpha} = \Omega(\sqrt{\log N})$.
\end{lemma}
\begin{proof}
    By \cite[Theorem~1.1]{Pyber97}, $\cG$ has an abelian subgroup $\cH$ of size $\geq 2^{\Omega(\sqrt{\log N})}$. By \cref{lem:abelianlb}, there exists a set $S\subset\cH$ of size $ = \Omega((\log|\cH|)/\log(4/\alpha))\geq\Omega((\sqrt{\log N})/\log(4/\alpha))$ such that for every $s\in S$, there exists a vector $v_s\in\R^{\cH}$ such that $\alpha v_s^\top (L_\cH)_sv_s > \sum_{t\in S\setminus\{s\}}v_s^\top (L_\cH)_tv_s$, where for any $h\in\cH$, $(L_\cH)_h\in\R^\cH$ denotes the Laplacian associated to $h$ in the ambient group $\cH$. Now, consider the vector $u_s\in\R^\cG$ which is defined as $(u_s)_x = (v_s)_x$ for $x\in\cH$, and $(u_s)_x = 0$ otherwise. Then note that for any $s, s'\in S\subset\cG$, $u_{s}^\top L_{s'}u_s = v_s^\top (L_\cH)_{s'}v_s$. Consequently, $\alpha u_s^\top L_su_s > \sum_{t\in S\setminus\{s\}}u_s^\top L_tu_s$ for all $s\in S\subset\cH\subset\cG$, thus showing that $\{L_s\}_{s\in S}$ is $\alpha$-minimal in $\mathcal{L}(G)$, as desired.
\end{proof}

Given that we can sparsify Cayley graphs up to $\polylog(N)$ factors, one may wonder if a similar sparsification result is possible for Schreier graphs. 

We show below that the answer is strongly negative.
\subsection{Lower Bounds against sparsifying graphs by sampling subgraphs}\label{lowerboundsparsgraph}
We first define the sparsification of graphs by sampling subgraphs. Indeed, let $G = G(V, E, w)$ be a (possibly weighted) graph, and let $E = \bigsqcup_{i = 1}^dH_i$ be a partition of $E(G)$ into subgraphs. Let $L_G$ be the Laplacian of $G$, and let $L_{H_i}$ be the Laplacian of $H_i$ for all $i\in[d]$. Note that $L_G = \sum_{i\in[d]}L_{H_i}$. 

We say that we are sparsifying $G$ by sampling the subgraphs $\{H_i\}$ if we have a map $\mu:[d]\to\R_{\geq 0}$ such that 
\begin{align}
\label{edgereducecond}
    \left|\bigcup_{i\in\supp(\mu)}H_i\right| < |\supp(w)|\mcom
\end{align}
and $\sum_{i\in [d]}\mu_iL_{H_i}\approx_\eps L_G = \sum_{i\in [d]}L_{H_i}$.

Note that graph sparsification (resp. Cayley sparsifiers of Cayley graphs) is a special case of the above setting, when $H_i$ are individual edges (resp. edges induced by a generator). We can similarly define Schreier sparsifiers of Schreier graphs. 

We shall sometimes abbreviate ``Cayley sparsifiers of Cayley graphs'', and ``Schreier sparsifiers of Schreier graphs'' as sparsifiers obtained by ``sampling generators''.

Before we state our lower bounds against sparsifying Schreier graphs, we mention a general consequence of \cref{thm:mainthmlowerbound}:
\begin{proposition}
\label{sparsificationlowerbounddel}
    Let $G = G(V, E)$ be a connected graph, and let $E = \bigsqcup_{i = 1}^dH_i$ be a partition of $E(G)$ into subgraphs such that $(V, E\setminus H_i)$ is not connected for all $i\in[d]$. Let $\cL:= \{L_{H_i}:i\in[d]\}$ be the collection of Laplacians of these subgraphs. Then $N(\alpha;\cL) = d + 1$ for all $\alpha > 0$, i.e.\ $\cL$ is $\alpha$-minimal for all $\alpha > 0$. In particular, by \cref{thm:mainthmlowerbound}, $\cL$ doesn't admit a $\eps$-sparsifier for any $\eps\in[0, 1)$, i.e.\ for any $\eps\in[0, 1)$ and any weight function $\mu:[d]\to\R_{\geq 0}$ such that $|\supp(\mu)| < d$, $(\{H_i\}_{i\in[d]}, \mu)$ doesn't $\eps$-sparsify $\cG$.
\end{proposition}
\begin{proof}
    As in the proof of \cref{lem:graphconnparam}, to show that $\cL$ is $\alpha$-minimal for all $\alpha > 0$, it suffices to show that for all $L_{H_i}\in\cL$, there exists a vector $x_i\in\R^V$ such that $x_i^*L_{H_i}x_i > 0 = \sum_{j\neq i}x_i^*L_{H_j}x_i$.

    Fix any $i\in[d]$. Since $(V, E\setminus H_i)$ is not connected, choose $S_i\subseteq V$ such that the edges of the cut $(S_i, V\setminus S_i)$ in $G$ is $H_i$. As in the proof of \cref{lem:graphconnparam}, if we set $x_i(v) = 1$ for all $v\in S_i$, and $-1$ otherwise, then $x_i^*L_{H_i}x_i > 0 = \sum_{j\neq i}x_i^*L_{H_j}x_i$, as desired.
\end{proof}

We can now proceed towards proving lower bounds for sparsifying Schreier graphs:
\subsubsection{Lower Bounds against Sparsifying Schreier Graphs}
The main result of this sub-subsection is the following:
\begin{lemma}
\label{schreiersparlowerbound}
    For any $n\geq 3$, there exists a Schreier graph $\Sch(\cG, S, X)$ on $|X| = 4n$ vertices with $|S| = n$ generators such that $\Sch(\cG, S, X)$ can't be $\eps$-sparsified by sampling generators for any $\eps\in[0, 1)$.
\end{lemma}
\begin{remark}
    In fact, the Schreier graph we construct above is a union of perfect matchings, where each matching corresponds to the edges induced by a generator. 
\end{remark}
Towards constructing such a Schreier graph, recall from the preliminaries the following lemma:
\restatelemma{matchingschreier}
More informally, any graph which is a union of perfect matchings can be viewed as a Schreier graph. Thus, in light of \cref{sparsificationlowerbounddel}, it suffices to construct a connected graph which is a union of perfect matchings, but which becomes disconnected on deleting any constituent matching.

Towards constructing such a graph, we make the following definition:
\begin{definition}[Transversely Partitionable Graphs]
     A graph $G = G(V, E)$ is called \emph{transversely partitionable} if there exists a perfect matching $M\subseteq E$, and matchings $M_1, \ldots, M_d$ on $V$ such that $E = M_1\sqcup\cdots\sqcup M_d$ and $|M\cap M_i| = 1$ for all $i\in[d]$. In this case, we call $M$ a \emph{transversal} of $M_1, \ldots, M_d$.
\end{definition}
We shall need an example of a transversely partitionable graph, which we provide in the lemma below:
\begin{lemma}
\torestate{
\label{knntranspar}
    The complete bipartite graph $K_{n, n}$ is transversely partitionable for all $n\geq 3$.
}
\end{lemma}
A stronger version of the above result was proved in \cite{Ger74}. We adapt their ideas to our context and present a simpler proof for the statement above in \cref{knntransparproof}.

We can finally prove \cref{schreiersparlowerbound}:
\begin{proof}[Proof of \cref{schreiersparlowerbound} using \cref{knntranspar,sparsificationlowerbounddel}]
    Let $G = G(V, E) = K_{n, n}$ be the complete bipartite graph. Since $G$ is transversely partitionable by \cref{knntranspar}, we can partition $E(G) = \bigsqcup_{i = 1}^n M_i$ into matchings, and let $M\subseteq E(G)$ be a perfect matching which is a transversal of $M_1, \ldots, M_n$.\footnote{In the case of $K_{n, n}$, since $M$ only has $n$ edges, all the matchings $M_1, \ldots, M_n$ have to be perfect} Let $M\cap M_i = \{\{a_i, b_i\}\}$ for all $i\in[n]$. Construct the graph $G' = G'(V', E')$, where 
    \[V':= V\sqcup\{a'_1, b'_1, a'_2, b'_2, \ldots, a'_n, b'_n\}\mcom\]
    and 
    \[E':= (E\setminus M)\sqcup\bigsqcup_{i = 1}^n\left(\{\{a_i, a'_i\}, \{b_i, b'_i\}\}\sqcup\bigsqcup_{j\neq i}\{\{a'_j, b'_j\}\}\right) = \bigsqcup_{i = 1}^nM'_i\mcom\]
    where $M'_i:= (M_i\setminus M)\sqcup\{\{a_i, a'_i\}, \{b_i, b'_i\}\}\sqcup\bigsqcup_{j\neq i}\{\{a'_j, b'_j\}\}$. Note that $M'_1, \ldots, M'_n$ are perfect matchings.\footnote{Note that $E'$ is not simple, as the edges $\{\{a'_i, b'_i\}:i\in[n]\}$ repeat $n - 1$ times.} 
    
    Now, since $(V, E\setminus M)$ is connected, $G'$ is also connected. However, note that for any $i\in[n]$, $(V', G'\setminus M_i)$ is not connected, since deleting $M_i$ separates $a'_i, b'_i$ from other vertices. Consequently, by \cref{sparsificationlowerbounddel}, $G'$ can't be $\eps$-sparsified by sampling the matchings $M'_1, \ldots, M'_n$.\footnote{Note that any sparsifying distribution $\mu$ must satisfy $|\supp(\mu)| < d$} However, since $E'$ is a union of perfect matchings, by \cref{matchingschreier}, $G'$ can be identified with a Schreier graph $\Sch(\cG, S, V')$, where $s = s^{-1}$ for all $s\in S$ (hence $S = S^{-1}$ is symmetric), and the matchings $M'_1, \ldots, M'_n$ correspond (uniquely) to the edges induced by the generators $s_1, \ldots, s_n\in S$. Consequently, $\Sch(\cG, S, V')$ can't be $\eps$-sparsified by sampling generators for any $\eps\in[0, 1)$, as desired.
\end{proof}

\section*{Acknowledgements}
We thank Akash Kumar for useful discussions and for highlighting \cite{GavvaZ24}, which explores graph splitting in Cayley graphs over cyclic groups. We thank \c{S}tefan Tudose for discussions during the project’s early stages. We also thank anonymous reviewers for pointing out typos and related work.

A.B.\ is especially grateful to Supravat Sarkar for numerous insightful exchanges and for directing us to \cite{Ger74}.
\bibliographystyle{alpha}
\bibliography{cayley}
\appendix
\section{Appendix}
\restatelemma{matchingschreier}
\begin{proof}\label{matchingschreierproof}
    We closely follow the proof of \cite[Theorem~2]{Gro77}.
    
    Note that any perfect matching $M$ on $V$ can be viewed as a permutation on $V$, with every $v\in V$ being sent to the $v'\in V$ it is matched to. Furthermore, if $M$ is an undirected graph, then the permutation it corresponds to has order $2$. Consequently we have permutations $\pi_1, \ldots, \pi_d\in\mathfrak{S}_V$ corresponding to the perfect matchings $M_1, \ldots, M_d$ such that $\pi_i^2 = \id$ for all $i\in[d]$, where $\mathfrak{S}_V$ is the group of permutations of $V$.

    Let $\cG = \langle \pi_1, \ldots, \pi_d\rangle\subseteq\mathfrak{S}_V$ be the group generated by the permutations $\pi_1, \ldots, \pi_d$. Note that $\cG$ acts transitively on $V$: Indeed, since $G$ is connected, for any $v\neq v'\in V$, we can find a path $(v =: v_0, v_1, \ldots, v_r, v_{r + 1} := v')$ in $G$ between $v$ and $v'$, and if the edge between $(v_i, v_{i + 1})$ belongs to $M_{j_i}$ for all $i\in\{0, 1, \ldots, r\}$, then note that $(\pi_{j_r}\cdots\pi_{j_0})(v) = v'$. Now, let $v_0\in V$ be some arbitrary element of $V$, and let $\cH:= \{\sigma\in \cG:\sigma(v_0) = v_0\}$. We claim $G\cong\Sch(\cG, \{\pi_1, \ldots, \pi_d\}, \cG\backslash\cH)$: Since $\cG$ acts transitively on $V$, for every $v\in V\setminus\{v_0\}$ we can choose $\sigma_v\in\cG$ such that $\sigma_v(v_0) = v$, and we set $\sigma_{v_0}:= \id_{\mathfrak{S}_V}$. Then it is easy to see that $V\ni v\mapsto \cH\sigma_v\in\cG\setminus\cH$ induces an isomorphism of $G$ and $\Sch(\cG, \{\pi_1, \ldots, \pi_d\}, \cG\backslash\cH)$.
\end{proof}
\restatelemma{knntranspar}
\begin{proof}\label{knntransparproof}
    First, assume $n$ is odd. Identify the parts of $K_{n, n}$ with (two copies of) $\Z_n = \{0, 1, \ldots, n - 1\}$, and let the matchings $M_0, \ldots, M_{n - 1}$ be $M_i:= \{(x, x + i):x\in\Z_n\}$. It is clear that these matchings partition the edges of $K_{n, n}$. Now, let $M$ be the matching $M:= \{(x, 2x):x\in\Z_n\}$. Note that $M$ is a matching because $2$ is an invertible element in $\Z_n$. Finally, $M\cap M_i = \{(i, 2i)\}$ for all $i\in\Z_n$, and thus we're done.

    Before we proceed to the proof for the even case, we note that in the above construction, $M':= \{(i + 1, 2i + 1):i\in\Z_n\}$ is also a transversal of $M_0, \ldots, M_{n - 1}$, with $M'\cap M_i = \{(i + 1, 2i + 1)\}$. Also note that $M'\cap M = \emptyset$.

    Now, assume $n\geq 4$ is even. Identify the parts of $K_{n, n}$ with (two copies of) $[n]$, and view $K_{n, n} = K_{n - 1, n - 1}\sqcup\{(n, i):i\in[n-1]\}\sqcup\{(i, n):i\in[n-1]\}\sqcup\{(n, n)\}$. Note that $n - 1$ is odd, and thus we can decompose $K_{n - 1, n - 1} = M_0\sqcup\cdots\sqcup M_{n - 2}$, and we also have two disjoint transversals $M, M'$ of $M_0, \ldots, M_{n - 2}$. For all $0\leq i\leq n - 2$ write $M\cap M_i = \{(a_i, b_i)\}$.
    
    Now construct matchings $T_0, \ldots, T_{n - 2}, T_{n - 1}$ on $K_{n, n}$ as follows: For $0\leq i\leq n - 2$, define $T_i:= (M_i\setminus M)\cup\{(n, b_i), (a_i, n)\}$. Finally, define $T_{n - 1}:= M\cup\{(n, n)\}$. Note that $K_{n, n} = T_0\sqcup\cdots\sqcup T_{n - 1}$. Also note that $M'\cup\{(n, n)\}$ is a perfect matching in $K_{n, n}$ which is a transversal of $T_0, \ldots, T_{n - 1}$, and thus we're done.
\end{proof}

\end{document}